\RequirePackage{rotating}
\documentclass{amsart}
\usepackage{fullpage}
\usepackage[table]{xcolor}
\usepackage{mathtools, amssymb, amsthm, subfig, rotating,booktabs,siunitx,multirow}
\usepackage{hyperref}
\usepackage{nicefrac}

\usepackage[abs]{overpic}

\theoremstyle{plain}
\newtheorem{lemma}{Lemma}
\newtheorem{definition}{Definition}

\definecolor{salmon}{RGB}{250,128,114}

\newcommand{\set}[1]{\left\{ #1 \right\}}
\newcommand{\paren}[1]{\left( #1 \right)}
\newcommand{\eqd}{\overset{\mathcal{D}}{=}}
\newcommand{\CL}[1]{\mathcal{G}\!\paren{#1}}
\newcommand{\prob}[1]{\mathbb{P}\!\paren{#1}}
\newcommand{\NPcomp}{$\mathcal{NP}$-\textrm{complete}}
\newcommand{\expect}[1]{\mathbb{E}\!\left[#1\right]}
\DeclareMathOperator{\VAR}{Var}
\newcommand{\Var}[1]{\VAR\!\paren{#1}}
\newcommand{\abs}[1]{\left| #1 \right|}

\newcommand{\R}{\mathbb{R}}

\DeclareMathOperator{\relhaus}{\mathcal{RH}} 
\DeclareMathOperator{\relhausR}{\overrightarrow{\mathcal{RH}}} 
\newcommand{\RH}[1]{\relhaus\!\paren{#1}} 
\newcommand{\RHd}[1]{\relhaus\!\paren{#1}} 
\newcommand{\RHrd}[1]{\relhausR\!\paren{#1}} 

\DeclareMathOperator{\editdist}{GED} 
\newcommand{\GED}[1]{\editdist\!\paren{#1}} 

\DeclareMathOperator{\kol}{\mathcal{KS}} 
\newcommand{\KS}[1]{\kol\!\paren{#1}} 

\title{Relative Hausdorff Distance for Network Analysis}
\author{Sinan G. Aksoy}
\author{Kathleen E.\ Nowak}
\address{Pacific Northwest National Laboratory, Richland, WA, 99352}
\email{sinan.aksoy@pnnl.gov, kathleen.nowak@pnnl.gov,stephen.young@pnnl.gov}

\author{Emilie Purvine}
\address{Pacific Northwest National Laboratory, Seattle, WA, 98109}
\email{emilie.purvine@pnnl.gov}
\author{Stephen J.\ Young}

\begin{document}

\begin{abstract}
Similarity measures are used extensively in machine learning and data science algorithms. The newly proposed graph Relative Hausdorff (RH) distance is a lightweight yet nuanced similarity measure for quantifying the closeness of two graphs. In this work we study the effectiveness of RH distance as a tool for detecting anomalies in time-evolving graph sequences. We apply RH to cyber data with given red team events, as well to synthetically generated sequences of graphs with planted attacks. In our experiments, the performance of RH distance is at times comparable, and sometimes superior, to graph edit distance in detecting anomalous phenomena. Our results suggest that in appropriate contexts, RH distance has advantages over more computationally intensive similarity measures.
%
\end{abstract}
\maketitle
\section{Introduction}

Similarity measures play a crucial role in many machine learning and data science algorithms such as image classification and segmentation, community detection, and recommender systems. A good deal of effort has gone into developing similarity measures for graphs, in particular, since they often provide a natural framework for representing unstructured data that accompanies many real-world applications. Some popular graph similarity measures currently used are graph edit distance \cite{Sanfeliu1983}, iterative vertex-neighborhood identification \cite{Blondel2004, Kleinberg1999}, and maximum common subgraph based distance \cite{Fernandez2001}. However, as graph datasets grow larger and more complex, the need for tools that can \textit{both} capture meaningful differences and scale well is becoming more critical. In this respect, a number of sophisticated yet costly graph similarity measures, such as those listed above, fall short.

The recently proposed graph Relative Hausdorff (RH) distance \cite{Simpson2015} is a promising measure for quantifying similarity between graphs via their degree distributions. Inspired by the Hausdorff metric from topology \cite{Hausdorff1914}, RH distance was devised to capture degree distribution closeness at all scales, and hence is well-suited for comparing the heavy-tailed degree distributions frequently exhibited by real-world graphs.  Furthermore, as recent work has shown \cite{RHLinear}, RH distance is extremely lightweight, with time complexity linear in the maximum degrees of the graphs being compared. However, as this metric is relatively new, it has not yet been extensively vetted. In particular, current research has not addressed its potential as an anomaly detection method for time-evolving graphs.

In this work, we conduct a statistical and experimental study of RH distance in the context of dynamic graphs. While RH distance may be applied to arbitrary pairs of networks, we focus our attention on sequences of time-evolving networks arising from cyber-security applications.  We begin by first applying RH distance to cyber-security logs recently released by Los Alamos National Laboratory, and investigate the extent to which it detects identified ``red team" events. Then we follow up by studying RH distance in the more general and controlled context of random dynamic graph models. Here we generate sequences of correlated Chung-Lu random graphs using a simplified cyber-security model proposed by Hagberg, Lemons, and Mishra \cite{Hagberg2016}, and test the extent to which RH distance detects several planted attack profiles. Throughout our analysis, we compare the performance of RH to that of more well-known graph similarity measures, such as edit distance and Kolmogorov-Smirnov distance of degree distributions. With this work, we better clarify the range of differences captured by RH, and also highlight its practical advantages and disadvantages over other methods.

\section{Preliminaries}
\begin{figure*}[t]
    \centering
      \includegraphics[scale=0.7]{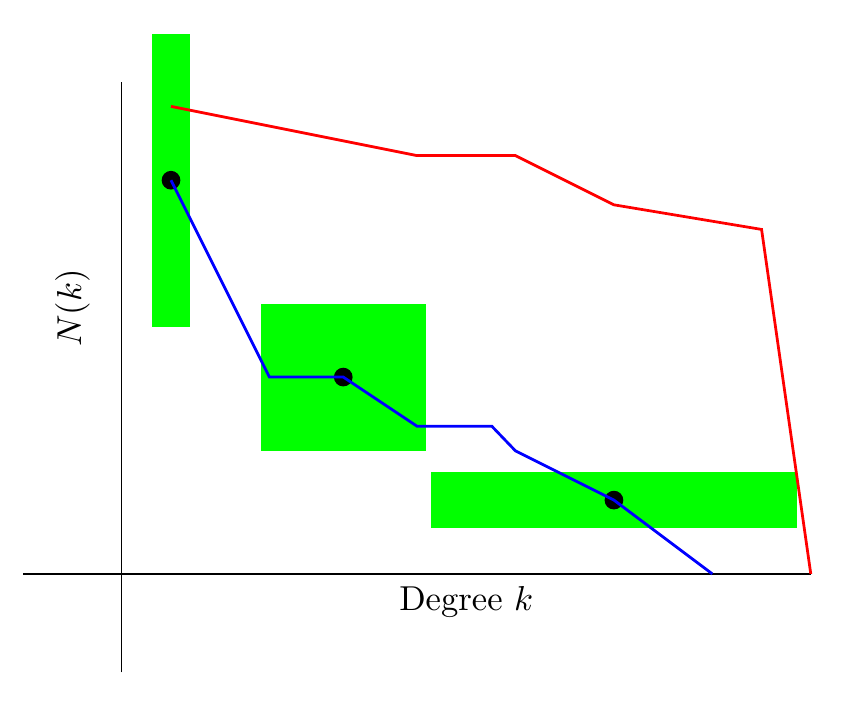}
        \caption{An illustration of the RH equivalent of an $\epsilon$-ball (in green) on one graph ccdh (in blue) compared with another (in red) at different points.} \label{fig:RHplot}
\end{figure*}

\subsection{Graph similarity measures}
Below we define the graph similarity measures we consider for anomaly detection in time-evolving graphs. We begin with the graph Relative Hausdorff distance, the primary focus of our study.
\subsubsection{Relative Hausdorff distance}
Originally introduced by Simpson, Seshadhri, and McGregor \cite{Simpson2015}, the Relative Hausdorff (RH) distance between graphs is a numerical measure of closeness between their complementary cumulative degree histograms (ccdh). More precisely, the (discrete) ccdh of a graph $G$ is defined as $\paren{N(k)}_{k=1}^{\infty}$, where $N(k)$ denotes the number of vertices of degree {\it at least} $k$. This is related to the commonly used {\it degree distribution}, which is defined as $(n(k))_{k=1}^{\infty}$ where $n(k)$ denotes the number of vertices with degree {\it exactly} $k$. Note that the ccdh and degree distribution are equivalent in the sense that each can be uniquely obtained from the other; nonetheless, for the purpose of this exposition, it is more convenient to work with the ccdh. Slightly abusing notation, we write $G(d)$ for a graph $G$ to mean the value of the ccdh of $G$ at $d$, and let $\Delta(G)$ denote the maximum degree of $G$. With these definitions in hand, the (discrete) RH distance between $F$ and $G$ is then defined as:

\begin{definition}[(Discrete) Relative Hausdorff distance \cite{Simpson2015}] \label{def:origRH}Let $F,G$ be graphs. The {\it discrete directional Relative Hausdorff distance} from $F$ to $G$, denoted $\RHrd{F,G}$, is the minimum $\epsilon$ such that
\begin{align*}
\forall d \in \{1,\dots, \Delta(F)\}, \exists d' \in \{1, \dots, \Delta(G)+1\} \mbox{ such that } |d-d'| \leq \epsilon d \mbox{ and } |F(d)-G(d')|\leq \epsilon F(d),
\end{align*}
and $\RHd{F,G}=\max\{\RHrd{F,G},\RHrd{G,F}\}$ is the {\it discrete Relative Hausdorff distance} between $F$ and $G$.
\end{definition}

In this paper, we compute RH distance using smoothed ccdhs, in which successive points are connected via line segments, as recommended by \cite{Matulef2017, Stolman2017}. Specifically, the authors define the smooth ccdh of a graph $G$, $G(d): \R_{\geq 1} \rightarrow \R_{\geq 0}$, as
$$
G(d) = \begin{cases}
\text{\# of vertices of degree at least } d,  &d \in \mathbb{Z}_{\geq 1} \\
(d-\lfloor d \rfloor)G(\lfloor d \rfloor)+(\lceil d\rceil-d)G(\lceil d \rceil), & d \in \mathbb{R}_{\geq 1}\setminus \mathbb{Z}.
\end{cases}
$$
In this case, the RH distance is defined much the same as before, except that the ccdh is piecewise linear. An illustration of the RH equivalent of an $\epsilon$-ball at points on a smooth ccdh is given in Figure \ref{fig:RHplot} and the precise definition of smooth RH distance is given below. Henceforth, we focus exclusively on smooth RH distance so we will drop the qualifier.

\begin{definition}[(Smooth) Relative Hausdorff distance \cite{Matulef2017, Stolman2017}] \label{def:continuousRH}Let $F,G$ be graphs. The {\it smooth directional Relative Hausdorff distance} from $F$ to $G$, denoted $\RHrd{F,G}$, is the minimum $\epsilon$ such that
\begin{align*}
\forall d \in \{1,\dots, \Delta(F)\}, \exists d' \in \mathbb{R}_{\geq 1} \mbox{ such that } |d-d'| \leq \epsilon d \mbox{ and } |F(d)-G(d')|\leq \epsilon F(d),
\end{align*}
and $\RHd{F,G}=\max\{\RHrd{F,G},\RHrd{G,F}\}$ is the {\it smooth Relative Hausdorff distance} between $F$ and $G$.
\end{definition}

By definition, $\RHd{F,G}=\epsilon$ means that for every degree $k$ in the graph $F$, $F(k)$ is within $\epsilon$-fractional error of $G(k')$ for some $k'$ within $\epsilon$-fractional error of $k$. Hence, the RH measure is flexible in accommodating some error in both vertex degree values as well as their respective counts, yet strict in requiring that {\it every point} in $F$ be $\epsilon$-close to $G$ (and vice versa).

Despite being referenced as a ``distance", RH distance does not satisfy the triangle inequality and hence is best viewed as a similarity measure rather than a bona-fide distance metric. However, to match existing literature we will still use the term ``RH distance.'' As shown in \cite{RHLinear}, $\RH{F,G}$ is extremely lightweight, and can be computed with run time $\mathcal{O}(\Delta(F)+\Delta(G))$. Lastly, we note that while RH distance values exceeding one have been  considered indicative of a ``large" dissimilarity between ccdhs, see \cite{Simpson2015}, it was shown in \cite{RHLinear} that the RH distance between graphs can be as large as $\mathcal{O}(m)$ when comparing graphs on $m$ and $n$ vertices with $m\geq n$.

\subsubsection{Other measures}

While RH distance will be the focus of the present work, we also consider several other graph similarity measures in order to provide relevant context for its performance.  First, we consider another comparable lightweight, ccdh-based measure called Kolmogorov-Smirnov (KS) distance. KS distance is a widely-used statistical measure of similarity between distributions, and serves as the test statistic for the two-sample KS hypothesis test \cite{Young1977}. In what follows, we will not only compute KS distance directly between graph degree distributions, but also between distributions of graph similarity values, such as RH values. To avoid confusion, below we define both KS distance as well as the two-sample KS hypothesis test (for which KS distance is a test statistic) for general empirical distributions.

\begin{definition}[KS distance and two-sample KS test]
Let $F \in \mathbb{R}^n$ and $G \in \mathbb{R}^m$ be empirical cumulative distribution functions. The Kolmogorov--Smirnov distance is
\[
\KS{F,G}=\max_x |F(x)-G(x)|.
\]
The null hypothesis that $F$ and $G$ are samples of two identical probability density functions is rejected at the $\alpha$ confidence level if
\[
\KS{F,G} > c(\alpha) \sqrt{\frac{n+m}{nm}},
\]
where $c(\alpha)=\sqrt{-\frac{1}{2} \log{\alpha}}$. The value $\exp({-\frac{2 \KS{F,G}^2 n m}{n+m} })$ is called the $p$-value of the two-sample KS test.
\end{definition}
For clarity, we emphasize that smaller values of KS distance indicate greater {\it similarity} between distributions, whereas smaller $p$-values permit one to reject the null hypothesis of identical underlying distributions at a higher confidence level, thereby presenting stronger evidence the empirical distributions were drawn from {\it different} underlying distributions. For the special case that $F$ and $G$ are the ccdhs of two graphs on $n$ and $m$ vertices, respectively, KS distance is given by $\KS{F,G}=\max_{x \in \mathbb{N}} |\widetilde{F}(x)-\widetilde{G}(x)|$, where $\widetilde{F}=\tfrac{1}{n}\cdot F$ and $\widetilde{G}=\tfrac{1}{m}\cdot G$ .
As argued in \cite{Matulef2017, Stolman2017}, KS distance between graph ccdhs can sometimes be large in graphs that are intuitively similar; furthermore, KS distance may also be insensitive to certain important differences between graph ccdhs, particularly in the tails of ccdhs (which correspond to high-degree vertices).

On the other end of the computational spectrum, we also consider {\it graph edit distance} (GED). Arguably one of the most well-known graph similarity measures, GED has been widely used throughout  machine learning, particularly in computer vision and pattern recognition contexts \cite{Gao2009}. An {\it edit operation} on a graph consists of either an {\it insertion}, {\it deletion}, or {\it substitution} of a single vertex or edge.\footnote{While others sometimes including {\it merging} and {\it splitting} edit operations, we restrict our attention to edit distance based on the three aforementioned operations.}
An {\it edit path of length $k$} between $F$ and $G$ is a sequence of edit operations $\mathcal{P}=(e_1,\dots,e_k)$ that takes $F$ to a graph that is isomorphic to $G$. Edit distance is the total weight of the minimum-cost edit path, i.e.

\begin{definition}[Graph edit distance]
Let $F,G$ be graphs. The {\it graph edit distance} between $F$ and $G$, denoted $\GED{F,G}$ is given by
\begin{align*}
\GED{F,G}&= \min_{\mathcal{P} \in \Upsilon(F,G)} \sum_{e_i \in \mathcal{P}} c(e_i),
\end{align*}
where $\Upsilon(F,G)$ denotes the set of possible edit paths from $F$ to $G$, and $c(e_i)\geq 0$ denotes a cost-function measuring the weight of edit operation $e_i$.
\end{definition}
In what follows, we simply take $c(e_i)=1$ for any edit operation, in which case $\GED{F,G}$ is the minimum number of edit operations needed to transform $F$ to $G$.

\subsection{Related Literature}

While in this work we focus on the graph Relative Hausdorff distance, we note that a wide variety of graph similarity measures have been utilized for anomaly detection in time-evolving graphs.
In \cite{ishibashi2010detecting}, the authors propose detecting anomalies in communication network traffic data by measuring cosine similarity between the principal eigenvectors of graph adjacency matrices.
In \cite{akoglu2010event}, Akoglu and Faloutsos also take an eigenvector-based approach for measuring graph anomalousness.
Matrix-analytic graph similarity measures have also been based on eigenvalue residuals \cite{Giuseppe2011}, non-negative matrix factorization \cite{Tong2011}, and tensor decompositions \cite{sapienza2015anomaly}.
Other popular approaches for graph-based anomaly detection are via {\it distance metrics}, such as those based on edit distance, maximum common subgraph distance, or mean vertex eccentricity \cite{Gaston2006UsingGD}, or take a community-detection approach towards identifying anomalies by tracking changes between clusters of well-connected vertices \cite{aggarwal2011outlier, Wang2017}.
For a broader survey of graph-based anomaly detection techniques see \cite{Akoglu2014, Ranshous2015, sensarma2015survey} and the references contained therein.
Lastly, we note that applications of graph similarity functions extend far beyond anomaly detection.
Graph similarity functions are also ubiquitous in inexact graph matching and graph classification problems.
For instance, graph edit distance is a key tool for error-tolerant pattern recognition and computer vision techniques~\cite{Gao2009}.
While in this work we explore Relative Hausdorff distance through the lens of anomaly detection, we note its application as a graph similarity measure in other contexts such as these remains unexplored.

As our focus in this paper will be on \emph{cyber} anomaly detection\footnote{Note that we do not consider signature-based methods like those employed in intrusion detection/prevention systems (e.g., Snort) to be anomaly detection methods. Instead these are rule-based behavior identification tools.} in network flow data, we also mention some of the existing graph-based methods specifically for cyber anomaly detection.
In this domain, some researchers focus on detecting anomalies edge-by-edge or target specific types of behavior, e.g., \cite{franccois2011bottrack,retina}, while others look at the graph more globally or structurally and are agnostic to the type of anomalous behavior being detected, e.g., \cite{ChPChSHeA2016}.
In \cite{retina} Noble and Adams describe a real-time unsupervised framework for detecting anomalies in network data. They consider ``edge activity'' as the sequence of flows on a single edge and compute correlations between event inter-arrival times and other edge data (e.g., byte count or protocol). Statistically significant changes in those correlations are flagged as anomalies. Groups of adjacent anomalies can be combined to form larger anomalies perhaps indicating coordinated behavior.
The authors of \cite{franccois2011bottrack} use PageRank to perform linkage analysis followed by clustering techniques to identify groups of IPs with similar behavior. These groups are then compared with known bot behavior to detect botnets within the network.
In the category of more structural and behavior-agnostic algorithms \cite{ChPChSHeA2016} introduces multi-centrality graph PCA and multi-centrality graph dictionary learning which use structural properties of a graph, e.g., walk statistics and centrality measures, to learn normal structure and thus detect abnormal structure. This method is not tailored to the cyber use case, but the authors use network flow as one of their examples.
Our work is similarly not targeted towards a specific cyber use case and is focused on detecting structural perturbations rather than clustering behavioral patterns.

\section{Los Alamos National Laboratory (LANL) Cybersecurity Data}

To begin our study of RH distance as an anomaly detection method for dynamic graphs, we will first consider a dataset recently released by LANL with known red team events from their internal corporate computer network \cite{Kent:LANL_RedTeamData, Kent:LANL_RedTeam}. The dataset represents 58 consecutive days of de-identified event data collected from four sources, namely:
\begin{itemize}
\item Windows-based authentication events from both individual computers and centralized active directory domain controller servers,
\item Process start and stop events from individual Windows computers,
\item Domain Name Service (DNS) lookups collected by internal DNS servers, and
\item Network flow data collected at several key router locations.
\end{itemize}
In total, the data set is approximately 87.4 gigabytes, spread across the four modalities, including 1,648,275,307 events coming from 12,425 users, 17,684 computers, and 62,974 processes. Ground truth for the red team events is given as a set of authentication events that are known red team compromise events. In this section, we will demonstrate that Relative Hausdorff distance is effectively able to identify anomalous behavior around the red team events in the LANL data.

\subsection{Data Source}
As stated above, LANL captured network evolution in four different modalities, namely authentication, process, network flow, and DNS events. In order to apply the RH distance, we first must convert these network event files into a time series of graphs. To do so, we consider 60 second moving windows that advance 20 seconds at a time. For each window we use the events in that window to construct a graph. For the 58 consecutive days, this yields a time sequence of 250,560 graphs for each modality. Further details for constructing each type of graph are given below.

\begin{itemize}
\item{\bf Authentication Graphs.}
The authentication data is a record of authentication events collected from individual Windows-based desktop computers, servers, and Active Directory servers. Each line of the data file reports a separate authentication event in the form
\begin{center}
 \texttt{time, sourceUser@domain, destUser@domain, source computer, dest computer,\\ auth type, logon type, auth orientation, pass/fail.}
\end{center}
For a given window, we construct an unweighted graph with edges $\{\text{sourceUser, destUser}\}$ for each user pair present in the logs within the window.
\item{\bf Authentication Failure Graphs.}
These are constructed in the same manner as the Authentication Graphs, except we restrict the edge set to those corresponding to failed authentications only.
\item{\bf Process Graphs.}
The process data is a record of process start and stop events collected from individual Windows-based desktop computers and servers. Each line of the data file reports a separate process start/stop in the form
\begin{center}
\texttt{
time, user@domain, computer, process name, start/end.}
\end{center}
For a given window, we construct an unweighted graph with edges $\{\text{computer, process name}\}$ for each computer-process pair present in the logs within the window.
\item{\bf DNS Graphs.}
The DNS data is a record of DNS lookup events collected from the central DNS servers within the network. Each line of the data file reports a separate lookup event in the form
\begin{center}
\texttt{
time, source computer, computer resolved,}
\end{center}
representing a DNS lookup at the given time by the source computer for the resolved computer. For a given window, we construct an unweighted graph with edges $\{\text{source computer, computer resolved}\}$ for each source-resolved computer pair present in the logs within the window.
\item{\bf Flow Graphs.}
The flow data is a record of the network flow events collected from central routers within the network. Each line of the data file reports a separate network flow event in the form
\begin{center} \texttt{
time, duration, source computer, source port, dest computer, dest port, protocol, \\ packet count, byte count.}
\end{center}
For a given window, we construct an unweighted graph with edges $\{\text{source computer, dest computer}\}$ for each source-destination computer pair that communicate during that time.
\end{itemize}

\subsection{Limitations}
While working with real-world data often presents challenges, testing graph-based anomaly detection methods on the LANL dataset is particularly difficult for several reasons. First and foremost, the data only provides red team authentication attempt time stamps and does not specify the nature, extent or duration of the red team events. This makes it difficult to segregate benign from anomalous time periods. Additionally, without knowing the specific red team actions, it is difficult to determine which (if any) of the aforementioned modalities a red team signature may appear in. Finally, it is worth noting the data exhibited large periods of time in which no events occurred that did not correspond to regular lulls such as weekends and nighttime. In particular, the flow data has records from only the first 37 of the 58 days.  To address some of these limitations, in Section \ref{sec:HLM} we extend our analyses to a generalized dynamic network model \cite{Hagberg2016} proposed by LANL scientists Hagberg, Mishra, and Lemons. While no synthetic model is a perfect substitute for real data, this model's conception and design was directly informed by direct access to the LANL cyber data \cite{kent2014anonymized, kent2016cyber} and provides a framework under which we may draw more certain and rigorous conclusions regarding the behavior of RH distance. First, we present our analysis of the real LANL data.

\subsection{Experiment and results}
As a first-pass approach towards studying the sensitivity of RH distances to red team events in the LANL dataset, we test whether the distribution of pairwise RH distance values before a red team event differs significantly from the post red team event distribution. In this way, we assess whether there is statistical evidence to support that red team events demarcate ``change-points" in RH distance distribution. To that end, for each red team event at time $r$, we associate a time window $w$ of length $\ell$ centered at $r$, which we denote $w_{\ell}(r)$. Each such window can be naturally partitioned into a ``before" period (i.e.\ the time interval $(r-\nicefrac{\ell}{2},r)$) and ``after" period, $(r, r+\nicefrac{\ell}{2})$. To avoid overlapping windows and ensure the ``before" periods are in fact devoid of red team events, we restrict attentions to windows in which no red team event occurs in the time interval $(r-\nicefrac{\ell}{2},r)$. Put equivalently, we consider the set
\[
W_\ell=\{w_\ell(r): \mbox{red team event occurs at $r$, no red team events occur in $(r-\nicefrac{\ell}{2},r)$}\}.
\]
We note that it is possible for the after period of a window in $W_\ell$ to contain additional red team events. For each window in $W_\ell$, we compute the RH distances between pairs of graphs separated by $\delta$ seconds in the before period, as well as such pairs belonging to the after period. We then aggregate the RH distances over all before periods and all after periods. More precisely, if $G_{0},G_{1},\dots$ denotes the time-ordered sequence of graphs for a particular mode in the LANL data, we compute the aggregate before and after distributions as
\begin{align*}
D_b &= \{\RH{G_{t},G_{t+\delta}}: t,t+\delta \in (r-\nicefrac{\ell}{2},r) \mbox{ and } w_\ell(r) \in W_\ell \},\\
D_a &= \{\RH{G_{t},G_{t+\delta}}: t,t+\delta \in (r, r+\nicefrac{\ell}{2}) \mbox{ and } w_\ell(r) \in W_\ell \},
\end{align*}
respectively. Recalling that we processed the LANL graph sequence for each modality by generating graphs for windows shifted by 20 seconds, we may choose the parameter $\delta$ controlling the granularity of pairwise RH measurements to be as small as $20$ seconds and and as large as  $\nicefrac{\ell}{2}-20$ seconds. Finally, we assess whether these aggregated before and after RH distance distributions differ significantly by conducting a two-sample Kolmogorov-Smirnov test. Table \ref{agg:pvals} presents the resulting $p$-values for $\delta=20,40,60,120,240$ seconds, under window lengths $\ell=30,60,120$ minutes, for each LANL modality.

\begin{table}[t]
\begin{tabular}{c||ccccc|ccccc|ccccc}
\toprule
& \multicolumn{5}{c}{Window: 30 min} & \multicolumn{5}{c}{Window: 60 min} & \multicolumn{5}{c}{Window: 120 min}\\
$\mbox{Mode/Shift}$& \multicolumn{1}{c}{20s} & \multicolumn{1}{c}{40s} & \multicolumn{1}{c}{60s} & \multicolumn{1}{c}{120s} & \multicolumn{1}{c}{240s} & \multicolumn{1}{c}{20s} & \multicolumn{1}{c}{40s} & \multicolumn{1}{c}{60s} & \multicolumn{1}{c}{120s} & \multicolumn{1}{c}{240s} & \multicolumn{1}{c}{20s} & \multicolumn{1}{c}{40s} & \multicolumn{1}{c}{60s} & \multicolumn{1}{c}{120s} & \multicolumn{1}{c}{240s}  \\ \midrule
\mbox{AuthFail} &0.20 & 0.49 & 0.70 & 0.97 & 0.57 & \cellcolor{green}0.01 & 0.39 & 0.10 & 0.76 & \cellcolor{yellow}0.02 & \cellcolor{yellow}0.02 & \cellcolor{green}0.01 & \cellcolor{yellow}0.02 & \cellcolor{green}0.01 & \cellcolor{green}0.00 \\
\mbox{Auth} &0.36 & 0.72 & 0.17 & 0.38 & 0.12 & 0.33 & \cellcolor{yellow}0.03& \cellcolor{green}0.01 & 0.63 & 0.06 & 0.35 & \cellcolor{green}0.00 & \cellcolor{green}0.00 & 0.36 & 0.07  \\
\mbox{Flow} &0.61 & 0.30 & 0.31 & 0.76 & 0.91 & 0.75 & 0.49 & 0.17 & 0.31 &  0.28 & \cellcolor{green}0.00 & \cellcolor{green}0.00 & \cellcolor{green}0.00 & \cellcolor{yellow}0.02 & \cellcolor{green}0.00  \\
\mbox{DNS} &0.40 & \cellcolor{yellow}0.02 & 0.41 & 0.77 & 0.12  & 0.97 & 0.36 & 0.60 & 0.86 & 0.16 & 0.59 & 0.55 & 0.13 & 0.34 & \cellcolor{yellow}0.04 \\
\mbox{Process} & 0.44 & 0.82 & 0.64 & 0.96 & 0.74 & \cellcolor{yellow}0.03 & 0.31 & 0.08 & 0.08 & \cellcolor{yellow}0.05 & 0.08 & 0.43 & 0.17 & 0.06 & \cellcolor{green}0.01 \\
\bottomrule
\end{tabular}
\caption{The $p$-values of the two-sample KS test comparing RH distance distributions of aggregated before and after periods of time windows centered at red team events in the LANL data. The $p$-values are rounded to two decimal places, with rounded values at most 0.01 highlighted in green and values between 0.02 and 0.05 highlighted in yellow.} \label{agg:pvals}
\end{table}

The $p$-values in Table \ref{agg:pvals} suggest that whether the aggregated distribution of RH values before red team events differs significantly from the post red team events depends crucially on the cyber modality, window length and granularity parameter $\delta$. In the case of a 30 minute window, almost none of the parameter settings for any modality result in statistical significance, while for a 2-hour window, a majority of parameter settings are significant at a level of 0.05. In this case, the before and after RH distance distributions over longer time windows surrounding red team events more frequently show significant differences, which is perhaps unsurprising. On the other hand, the changes in significance levels as the granularity parameter $\delta$ varies are more difficult to interpret. Even for a fixed window length and modality, the significance levels neither consistently increase nor decrease in $\delta$.

One plausible hypothesis for this experiments sensitivity to $\delta$ is that RH distance values exhibit periodic behavior both within and across modalities, reflecting the natural circadian rhythms one might expect from temporal cyber data. If this were the case, the choice of $\delta$ may skew the RH values sampled when constructing the representative before and after distributions. To check whether such periodicity is indeed present in the RH distance measurements on LANL, we constructed heatmaps of RH distances between all pairs of graphs over given time windows. As this requires a quadratic number of comparisons, it is worth noting this analysis is crucially facilitated by the lightweight computational complexity of RH distance. We examined heatmaps not only for windows surrounding anomalies, but also for time windows away from red team events. Figure \ref{fig:heat} (left column) presents sample heatmaps for the Authentication, Flow and Process modalities spanning a 2-hour time period.  In an effort to select a representative window for short-term nominal RH behavior within each modality, this time period was selected so as to not include any red team events nor be preceded or followed by any red team events for 20 hours.  
It is also worth pointing out that the RH distance between pairs of flow graphs regularly exceeds one, indicating that the rough guide for detecting anomalous behavior given in \cite{Simpson2015} is inappropriate for the cyber-security context.
We also transformed each heatmap of pairwise RH distance values into a similarity matrix by applying the Gaussian kernel with $\sigma=1$, and performed normalized Laplacian spectral clustering\footnote{The prescribed number of clusters was chosen to coincide with the first observed gap in Laplacian eigenvalues, as in \cite{Luxburg2007}.}, as described by Ng, Jordan and Weiss \cite{ng2002spectral}. Under the corresponding heatmap, Figure \ref{fig:heat} (right column) plots the pairs of graphs belonging to common clusters, using a different color for each cluster (and white for different clusters).

\begin{figure}[t]
  \centering
  \hfill
  \subfloat[][Authentication Heatmap]{\includegraphics[trim = 28 190 32 210, clip, width = 0.4\textwidth]{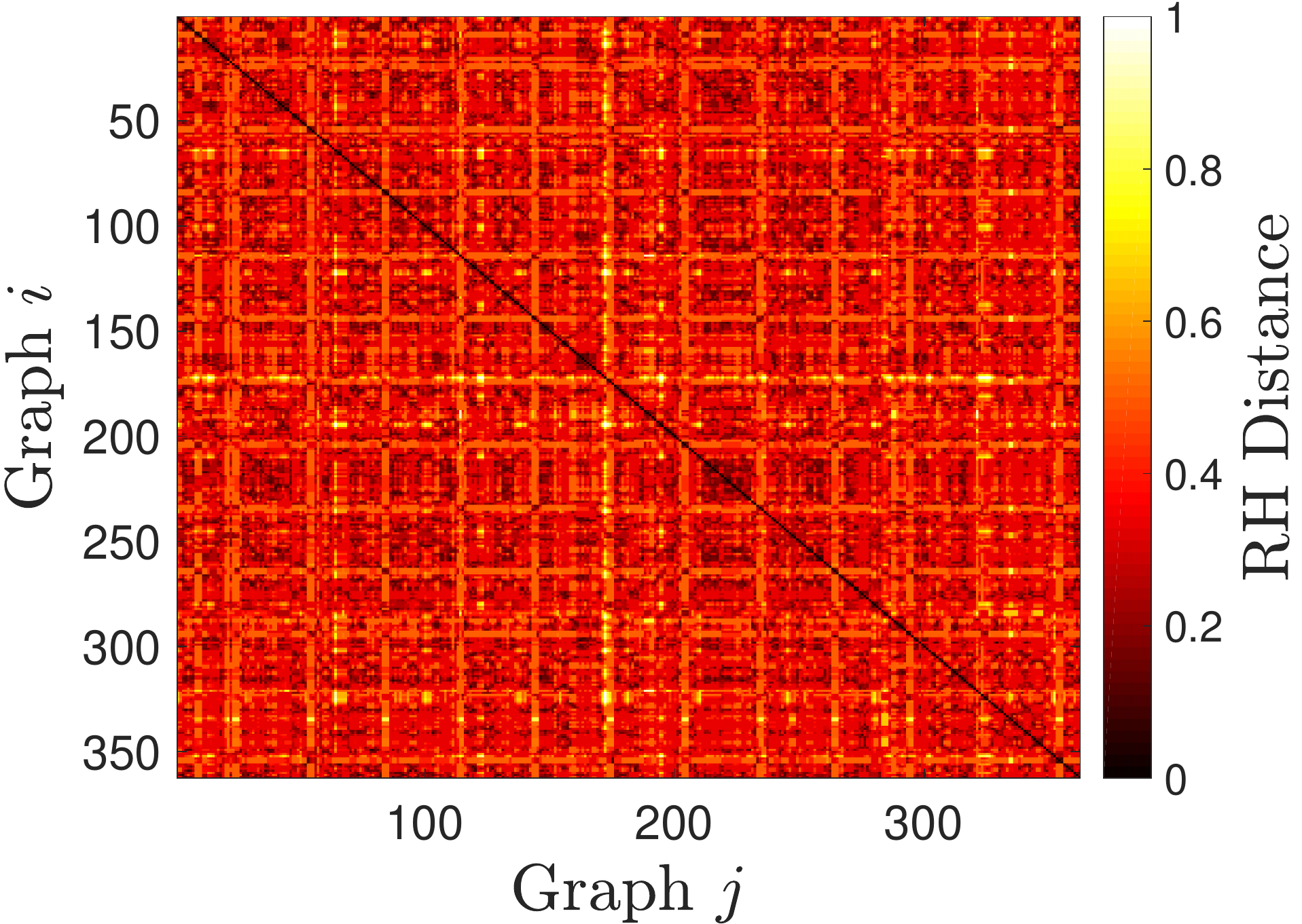}}
  \hfill
  \subfloat[][Authentication Clustering]{\includegraphics[trim = 28 190 32 210, clip, width = 0.4\textwidth]{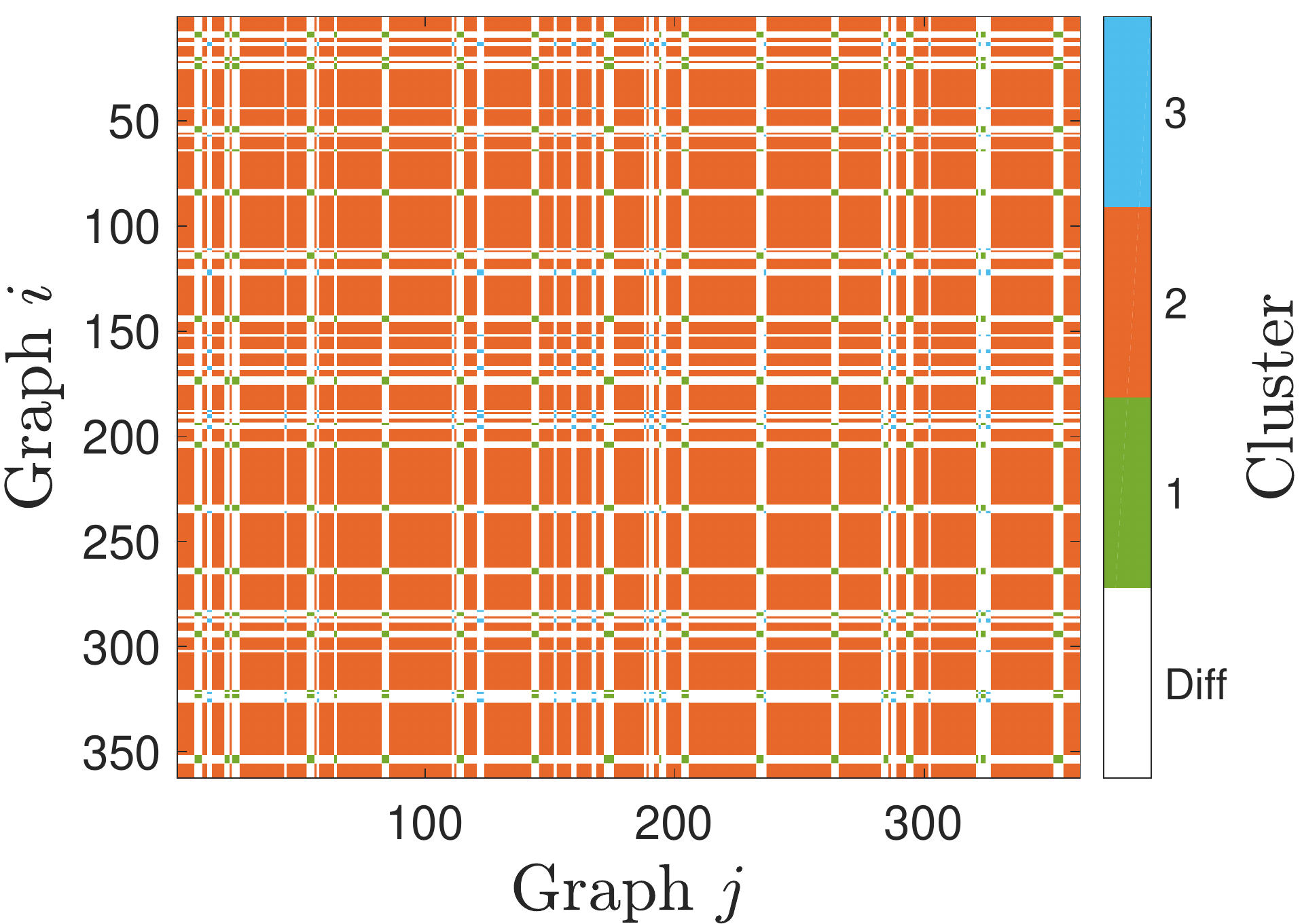}}
  \hfill \phantom{} \\
  \hfill
  \subfloat[][Flow Heatmap]{\includegraphics[trim =28 190 32 210, clip, width = 0.4\textwidth]{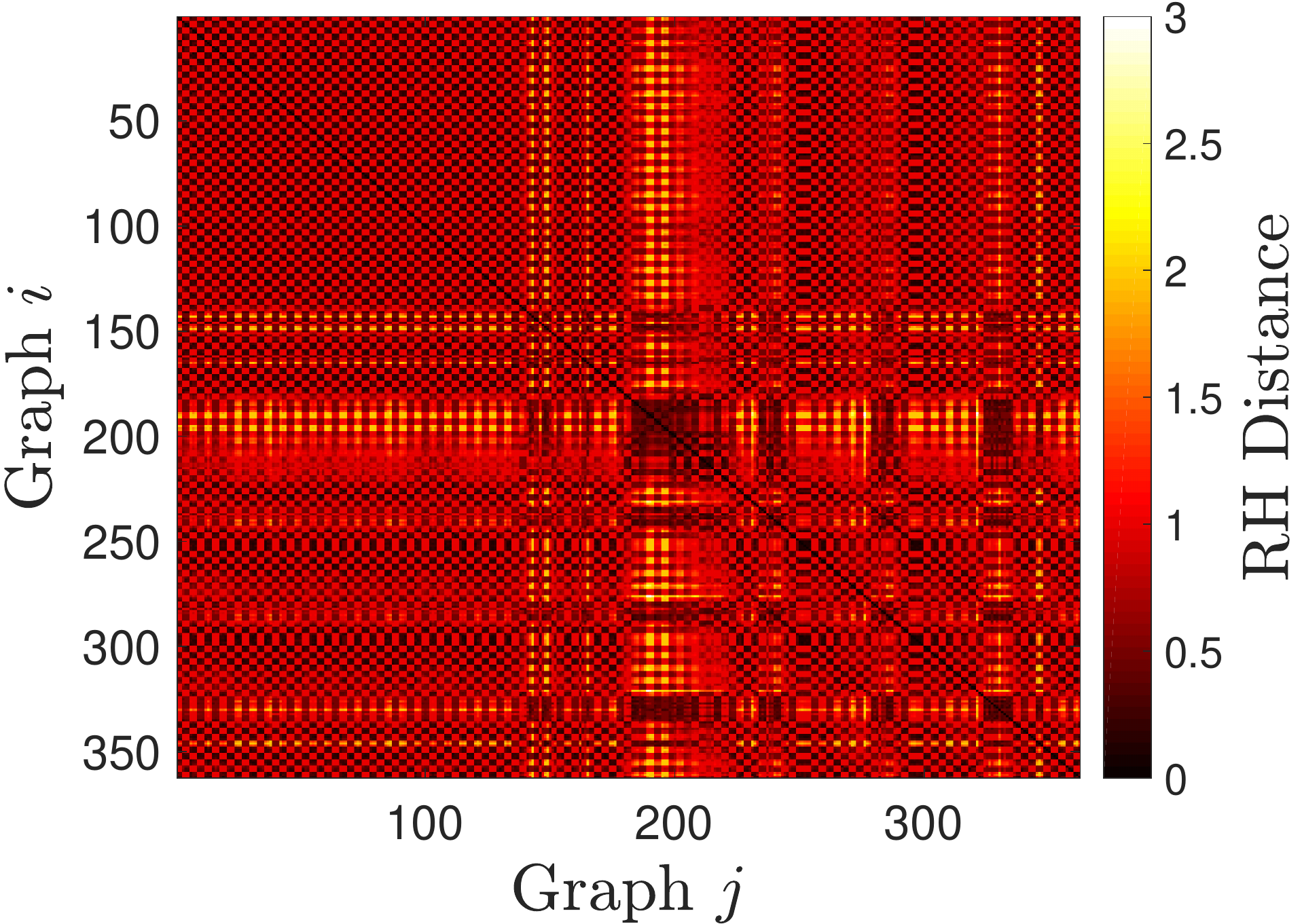}}
  \hfill
  \subfloat[][Flow Clustering]{\includegraphics[trim =28 190 32 210, clip, width = 0.4\textwidth]{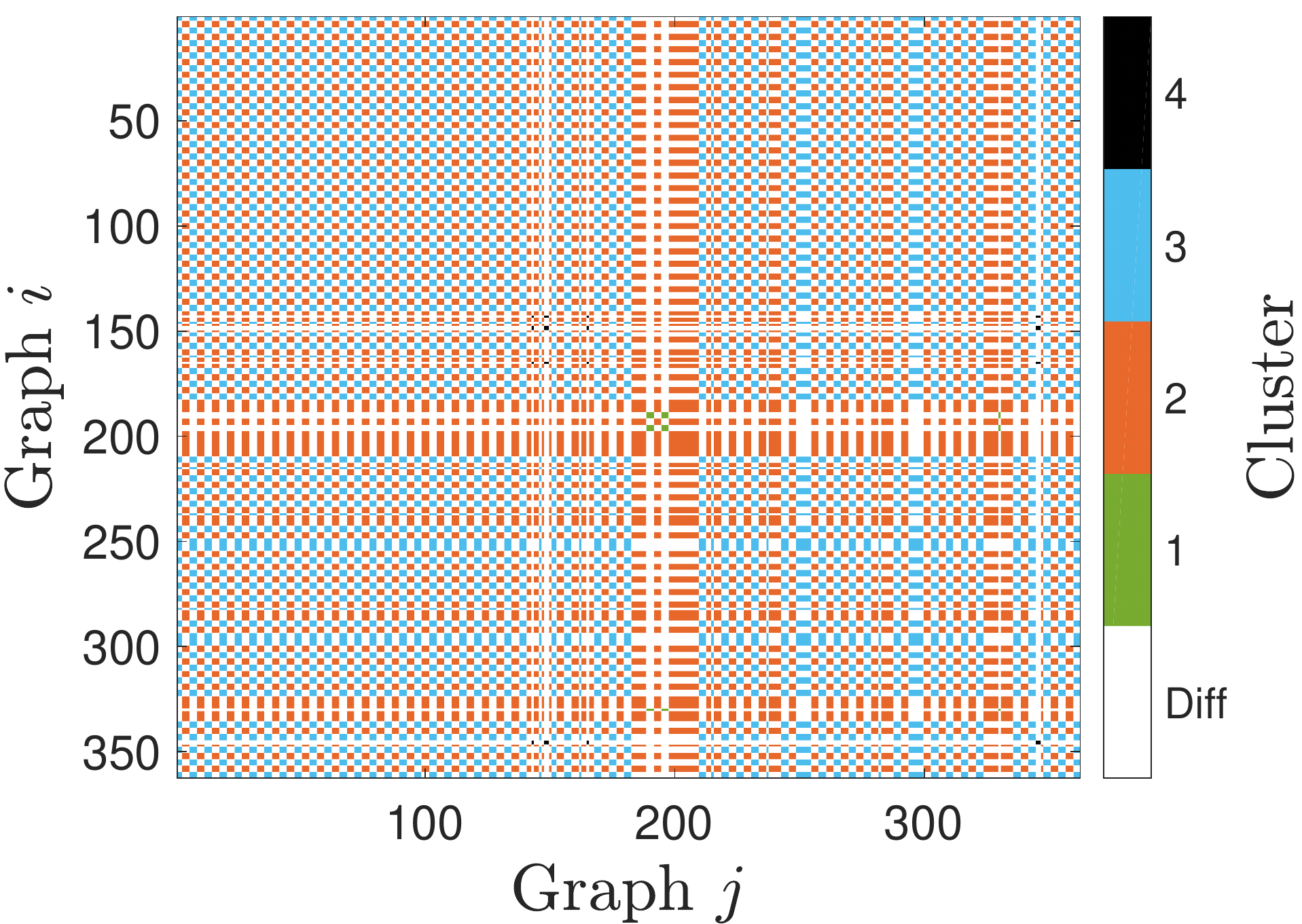}}
  \hfill \phantom{} \\
  \hfill
  \subfloat[][Process Heatmap]{\includegraphics[trim =28 190 32 210, clip, width = 0.4\textwidth]{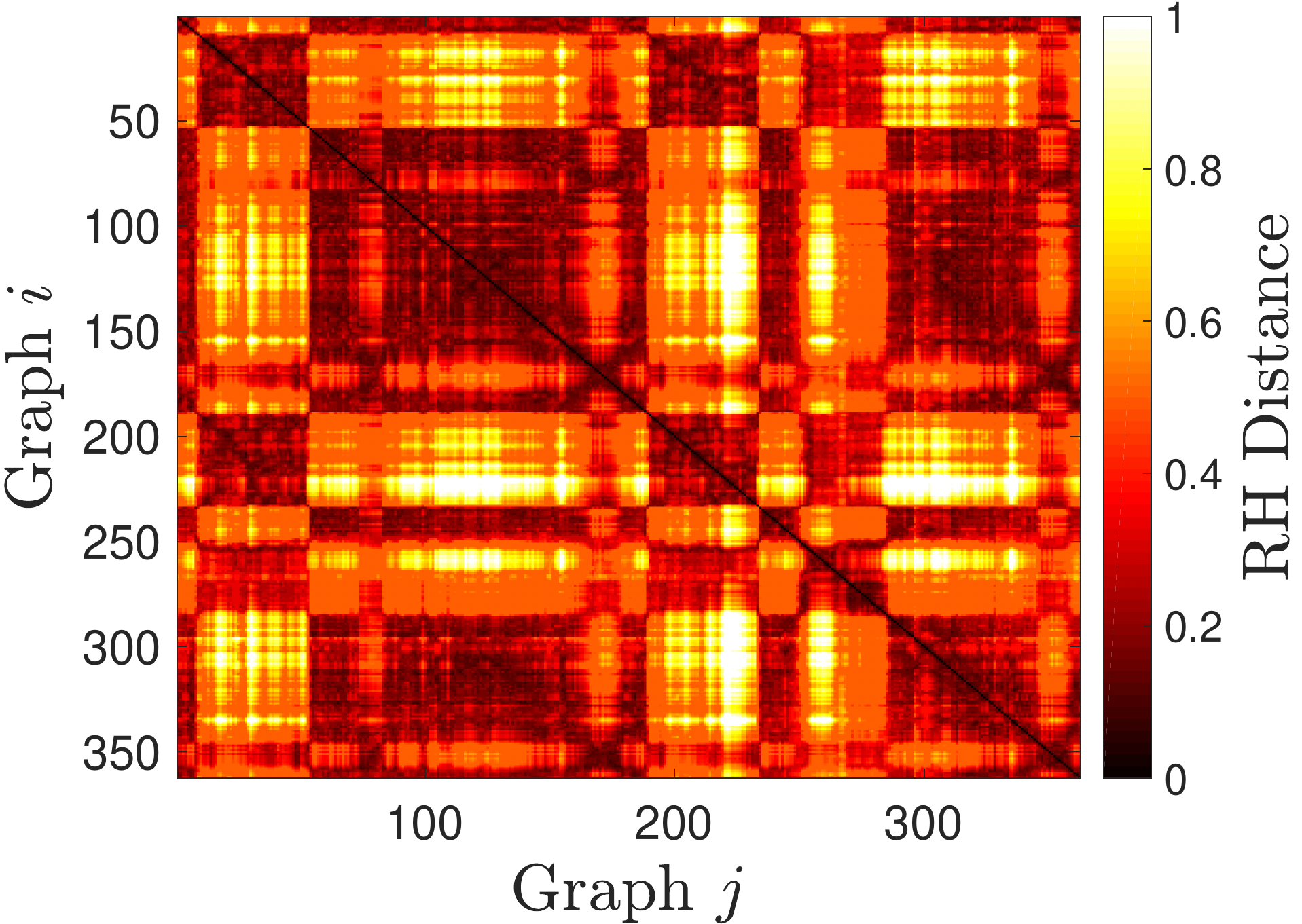}}
  \hfill
  \subfloat[][Process Clustering]{\includegraphics[trim =28 190 32 210, clip,width = 0.4\textwidth]{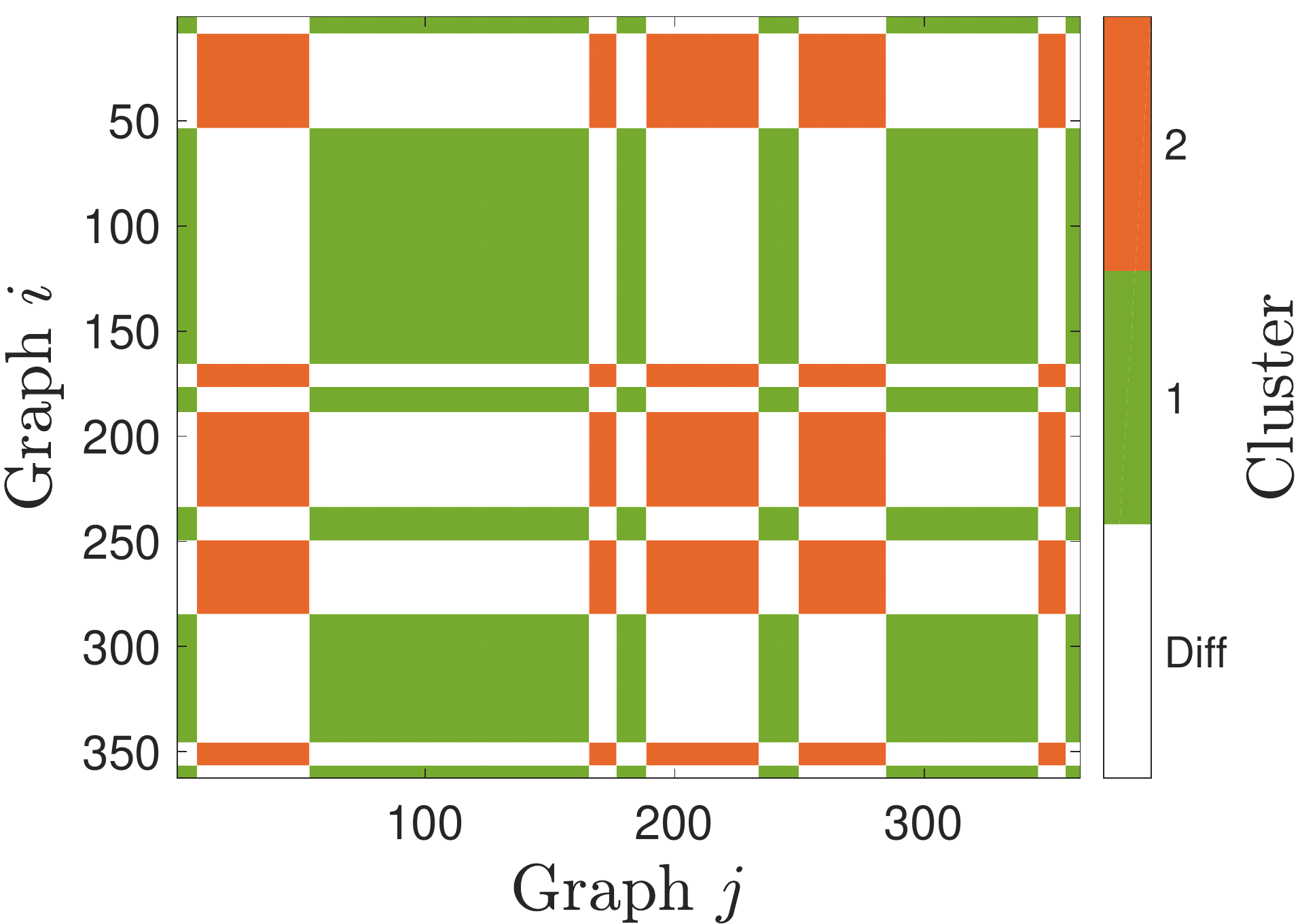}}
  \hfill \phantom{} \\
\caption{\emph{Left Column:} RH distance heatmaps  for nominal two-hour period. \\ \phantom{Figure 3.\quad} \emph{Right Column:} Spectral clustering of similar graphs based on RH distances.}\label{fig:heat}
\end{figure}

A cursory examination of the heatmaps and their clustering suggests that the RH distance values for a given modality exhibit persistent and striking periodic patterns.
Furthermore, in comparing the plots for Flow, Authentication, and Process, the differences in the periodic behavior of RH values are also apparent.
While these periodicities are not entirely unexpected they are likely network- and data-dependent.
Detecting and visualizing periodic behavior in network data is an active area of research, e.g., \cite{vizPeriodicity,hubballi2013flowsummary,periodicityLANL}.
As a consequence of this experiment and examination of the heatmaps in Figure \ref{fig:heat}, it is clear that a single choice of granularity parameter $\delta$ is likely insufficient in establishing a representative distribution of RH values within a time window for any given modality. Accordingly, next we refine our experiment to better account for the inherent {\it multi-scale} and {\it multi-modal} nature of the LANL data.

One of the many difficulties with investigating the effectiveness of RH distance in detecting anomalous behavior associated with red team events in the LANL data sets is that the red team process is inherently multi-modal and multi-scale.  That is, the red team events identified in the data are simply the first step of the red team intrusion process which could potentially affect all of the data modalities (process, DNS, flow, and authentication) and occur over multiple time scales.  To attempt to deal with this issue, we craft an indicator for each time that considers the RH distance between graphs with multiple time differences and in multiple modalities.  More concretely, let $\mathcal{S}$ be the set of potential data sources and let $\mathcal{G} = \set{G_{t,s}}$ be the collection of observed graphs indexed by the time $t$ and data source $s \in \mathcal{S}$.  For any fixed timestamp $t$ and collection of differences $\mathcal{D}$, we will define the \emph{profile vector} at time $t$, $v^{(t)}$, as the vector given by $\paren{\RH{G_{t,s},G_{t-\delta,s}}}_{s\in\mathcal{S},\delta \in \mathcal{D}}$.  Ideally, this profile could be used to aggregate the behavior across multiple modalities and multiple time scales and give a clearer picture of the overall state of system.  However, there is a further complication with this approach in that the LANL data represents a system which has a naturally evolving behavior based on various temporal patterns of human activity (i.e.\ weekday vs. weeknight, circadian rhythms, etc.).  To adjust for these temporal patterns, for every time $t$, source $s \in \mathcal{S}$, and difference $\delta \in \mathcal{D}$, we define a baseline behavior random variable $\mathcal{B}_{t,s,\delta}$ which is the random variable which represents the ``typical'' behavior of $\RH{G_{t,s},G_{t-\delta,s}}$.  This baseline behavior can then be combined with the profile vector $v^{(t)}$ to generate a \emph{temporal profile vector} $\hat{v}^{(t)}$ where for any $(s,\delta) \in \mathcal{S} \times \mathcal{D}$ we have $\hat{v}^{(t)}_{s,\delta} = \prob{ v^{(t)}_{s,\delta} - \epsilon < \mathcal{B}_{t,s,\delta} < v^{(t)}_{s,\delta} + \epsilon }$.  We define the \emph{temporal score} of the time $t$ as the geometric mean of the entries of $\hat{v}^{(t)}$
\footnote{As a practical matter, for those entries $(t,s,\delta)$ where there is insufficient or no data to estimate $\hat{v}^{(t)}_{s,\delta}$, the entry is dropped from the vector and ignored in the calculation of the temporal score.}.  In what follows, we will calculate the temporal scores of time periods before and after a red team authentication event and show that there is a statistically significant difference between the behaviors.  In fact, we will show that this temporal scoring methodology is more sensitive than using raw RH scores evaluated in Figure \ref{agg:pvals} indicating that there are significant potential gains to be found by considering multi-modal and multi-scale indicators for anomalous behaviors.

Before applying our results to the LANL data sets, it remains to address how to estimate the distribution of the random variable $\mathcal{B}_{t,s,\delta}$ and how the $\epsilon$ term defines the temporal profile vector is chosen.  For a fixed $t$ we estimate the empirical distribution for $\mathcal{B}_{t,s,\delta}$, we consider the RH distance between all pairs of graphs $\paren{G_{t^*,s},G_{t^{*}-\delta,s}}$ where $t^{*}$ ranges over all times that differ from $t$ by a multiple of a week, plus or minus 10 minutes.  In order to avoid biasing this empirical estimate we exclude times $t^{*}$ where there is a red team event in the interval $[t^{*} - \delta, t^{*}]$ as well as those that are within 10 minutes of $t$.  As we see in Figure \ref{fig:heat} the typical variation of a RH distance changes significantly based on the modality of the observation, both in source and elapsed time between graphs.  Thus, rather than fixing a particular value of $\epsilon$, we choose $\epsilon$ as one twentieth of the range of the empirical distribution for $\mathcal{B}_{t,s,\delta}$.  Finally, for each of the 712 red team times provided in the LANL data we calculate the temporal scores for each graph in the 30 minutes before and after the red team time and apply the two-sample Kolmogorov-Smirnov test, see Table \ref{T:temporal}. We further segregate this data by whether or not additional red team events occur during the 30 minutes prior to the red team time.

\begin{table}
  \begin{tabular}{c | r | r | r || r}
    & {\centering $p \leq 0.10$} & {\centering $p \leq 0.05$} & {\centering $p \leq 0.01$} & {\centering Total} \\ \midrule
    intervals with no prior red team events & 30 & 30 & 21 & 48 \\
    intervals with prior red team events & 478 & 460  & 370 & 664 \\
  \end{tabular}
  \caption{Aggregate behavior of temporal scoring on a per event basis}\label{T:temporal}
  \end{table}

It is clear from Table \ref{T:temporal} that the temporal scores is far from a perfect indicator, as a non-negligible fraction of the changes associated with a red team event are not detected. Nonetheless, it is also apparent that for a relatively lightweight measure the RH distance exhibits reasonable effectiveness in distinguishing between nominal and anomalous behaviors.  However, our conclusions must be somewhat tempered by the challenging nature of real world data and the LANL data in particular.  Specifically, the lack of clear demarkation between anomalous and non-anomalous behavior as well as the limited time-scope of the investigation are significant caveats to any conclusions we make about the effectiveness of RH distance.  In following section, we attempt to address these caveats by analyzing a synthetic temporal graph model inspired by the LANL cyber data.

\section{Simulated Evolving Networks} \label{sec:HLM}
The study of temporal networks is concerned with the analysis and modeling of time-ordered sequences of graphs. In order to better understand temporal network dynamics, researchers have proposed a plethora of abstract models for their simulation (for a survey, see \cite{Holme2012}). In the present work, we consider a temporal graph model that belongs to the broader class of Markovian Evolving Graphs (MEGs) \cite{Avin2008}. Given a probability distribution over the set of all graphs on a fixed vertex set, MEGs have the defining property that the distribution at time $t$ is completely determined by that at $t-1$, thereby forming a sequence of random variables which satisfy the Markov property. Because of their generality and flexibility, MEGs have been popularly used to study {\it information spreading processes}, such as file sharing on peer-to-peer networks, social network memes, and disease spreading \cite{Clementi2014, Clementi2010}.

In \cite{Hagberg2016}, Hagberg, Lemons and Mishra proposed a new MEG model, the design of which was informed by their study of LANL centralized authentication system cyber data \cite{kent2014anonymized, kent2016cyber}. In particular, they observed that these sequences of graphs exhibit certain stable global properties, such as skewed degree distributions, while local dynamics such as individual vertex neighborhoods change rapidly. To capture these dynamics, they designed a temporal model that can be used to preserve certain random graph structure while affording tunable control over the rate of dynamics. We refer to their model as the HLM model. Ultimately, Hagberg et.\ al utilized the HLM model to study {\it temporal reachability}; that is, the expected time (number of evolutions or transitions) before a constant fraction of the vertices are reachable from an arbitrary vertex. We note that although the HLM model was developed to capture abstract dynamics exhibited by cyber data, the HLM model need not be limited to simulating cyber phenomena.  Although, the experiment that follows is driven by cyber-security structures and data, the is no \emph{a priori} reason that a similar experiment could not be applied across the variety of domains for which the evolving nature of the HLM model is appropriate, such as communication networks, social networks, and (on a much slower time scale) transportation networks.
In the remainder of this section, we study the sensitivity of RH distance in detecting several planted attack profiles, utilizing the HLM model to simulate the natural time evolution of a generic cyber network graph topology. Before describing our experimental methodology, we first begin by defining and briefly discussing the HLM model.

\subsection{Hagberg-Lemons-Mishra (HLM) Model}
As the HLM model can be viewed as a time-evolving generalization of
the Chung-Lu model $\CL{w}$, we will first briefly review the Chung-Lu
model as introduced in
\cite{Chung:ExpectedDegAverageDist,Chung:ExpectedDegreeAverageDist}.
The parameterization vector of the Chung-Lu model, $w$, is
$n$-dimensional where $n$ is the number of vertices in the graph.
Additionally, the vector $w$ satisfies that $w_v \leq \sqrt{\rho}$ for all $v$
where $\rho = \sum_{i=1}^n w_i$.  From the parameter $w$ the Chung-Lu
model is generated by including each edge $\set{u,v}$, independently, with
probability $\nicefrac{w_uw_v}{\rho}$.   For overview of many of the
known properties of the Chung-Lu model see the recent monograph
\cite{Chung:ComplexNetworks}.

The HLM model generates an infinite sequence of graphs $G_0, G_1, G_2,
\ldots$ with the property that there is a fixed vector $w$ such that for all $i$, $G_i \eqd \CL{w}$ where $\eqd$ is
equality in distribution.  In order to generate this sequence an
additional parameter, $\alpha$, is introduced to tune the extent to
which graph $G_{i+1}$ is controlled by $G_i$.   Specifically, $\alpha
\in [0,1]^n$ and $G_{i+1}$ is formed from $G_i$ by generating a masking
set $M$ where each pair $\set{u,v}$ is in $M$ independently with
probability $\sqrt{\alpha_u\alpha_v}$.  For an edge $\set{u,v} \not\in
M$, $\set{u,v} \in G_{i+1}$ if and only if $\set{u,v} \in G_i$, while
each potential edge $\set{u,v}$ in $M$ is present independently with
probability $\nicefrac{w_uw_v}{\rho}$.  In summary, we have that
\[\prob{\set{u,v} \in G_{i+1}} = \begin{cases} 1 -
  \sqrt{\alpha_u\alpha_v} +
  \sqrt{\alpha_u\alpha_v}
  \nicefrac{w_uw_v}{\rho}& \set{u,v} \in G_i
  \\
  \sqrt{\alpha_u\alpha_v}\nicefrac{w_uw_v}{\rho} & \set{u,v} \not\in G_i\end{cases}.\]
The fact that $G_{i+1} \eqd G_i$ follows by induction and the
observation that
\[ \nicefrac{w_uw_v}{\rho}\paren{1 -
  \sqrt{\alpha_u\alpha_v} +
  \sqrt{\alpha_u\alpha_v}
  \nicefrac{w_uw_v}{\rho}} +
\paren{1-\nicefrac{w_uw_v}{\rho}}\sqrt{\alpha_u\alpha_v}\nicefrac{w_uw_v}{\rho}
= \nicefrac{w_uw_v}{\rho}.\]

We note that there is a natural trivial generalization of the HLM
model where the edge probability $\nicefrac{w_uw_v}{\rho}$ is replaced
with arbitrary values in $p_{uv} \in [0,1]$.  In this case, at each
time step the network is distributed over graphs like
$\mathcal{G}(P)$, the generic independent edge graph model with
parameter $P$.  Similarly, the evolution parameter $\alpha$ can be
generalized to a symmetric matrix $A \in [0,1]^{n\times n}$.  We note that several well studied models fall into
this framework, including the stochastic block model, stochastic Kronecker graphs~\cite{Leskovec2005, mahdian2007stochastic}, random dot
product graphs~\cite{young2008random, Young2008, Younga}, and the inhomogeneous random graph
model~\cite{Bollobas2007, Soederberg2002}.  In order to maintain consistent notation, we will specify all
of the experiments in this work in terms of this generalized HLM model
even though most of generative matrices $P$ come from the Chung-Lu
model.  Further, with the aim of having the minimum number of
free-parameters we will only consider HLM evolutions where
$A_{uv}= A_{xy}$ for all $u \neq v$ and $x \neq y$.  We will further
slightly abuse notation and refer to this common value as $\alpha$.

Finally, we note that this generalized framework can be further
expanded by allowing the parameter matrix $P$ to depend on the time
step $t$.  In particular, we have
\[ \prob{\set{u,v} \in G_{t+1} \mid G_t} = \begin{cases} (1-\alpha) +
    \alpha p^{(t+1)}_{uv} & \set{u,v} \in G_t \\
    \alpha p^{(t+1)} & \set{u,v} \notin G_t\end{cases}.\]
It is worth mentioning that in this case $G_{t+1}$ is not distributed
like $\mathcal{G}(P^{(t+1)})$ because of the possibility of edges
being present from earlier timesteps.  In fact, it is an easy
exercise to show that the edges of $G_t$ are distributed according
to
$(1-\alpha)^tP^{(0)} + \sum_{i=1}^t \alpha (1-\alpha)^{t-i}P^{(i)}.$

\subsection{Experimental Setup}
In our experimental setup, in keeping with the lightweight nature of the RH calculation, we focus on the detection of small anomalies in extremely sparse graphs, such as we observed in small time windows for the LANL data set and other proprietary network flow data.
For the sparse graphs we consider two different fixed degree distributions. Both of these degree distributions are formed by choosing 5000 samples from some fixed probability distribution.  For the first degree distribution, we estimate a degree density function using a smoothed median estimator from a selection of one minute graphs in the LANL network flow data set, see Figure \ref{F:degree_pl}.  The resulting degree density function results in a ``power-law'' like degree distribution with exponent approximately $3.5$.  Although the resulting degree density function is not truly a power-law distribution, we will abuse notation and refer to it as a ``power-law'' degree distribution.  For a discussion of the difficulties and appropriateness of the power-law degree distribution for real data the interested reader is referred to the recent work \cite{broido2018scale}.   The resulting distribution has 4,742 edges in expectation as well as maximum expected degree 961.    The second distribution represents what we call a ``bump power-law,'' that is, a power-law distribution coupled with an approximately binomially distributed ``bump'' at higher degrees.  This can be thought of as a more hub-and-spoke style network where the degree of the spoke vertices are approximately power-law distributed while the degree of the hub vertices are approximately binomially distributed.   For the bump power-law distribution, the degree probabilities were explicitly estimated from a collection of several thousand graphs generated from a proprietary enterprise boundary network flow data set  (see Figure \ref{F:degree_plb}).  This resulting distribution has 6,067 edges in expectation as well as maximum expected degree 327.\footnote{It is worth mentioning that both of these degree distributions violate the standard assumption for the Chung-Lu model $\max_v w_v \leq \sqrt{\rho}.$    To deal with this, we replace that edge probabilities of $\nicefrac{w_uw_v}{\rho}$ with $\min\set{1, \nicefrac{w_uw_v}{\rho}}$.  However, as there are under 200 pairs $\set{u,v}$ where $w_uw_v > \rho$ for each of the degree distributions, this makes a minimal difference in the model.}

\begin{figure}
  \centering
  \hfill
  \subfloat[][Power-Law\label{F:degree_pl}]{
    \includegraphics[trim = 18 33 71 76, clip, width = .4\textwidth]{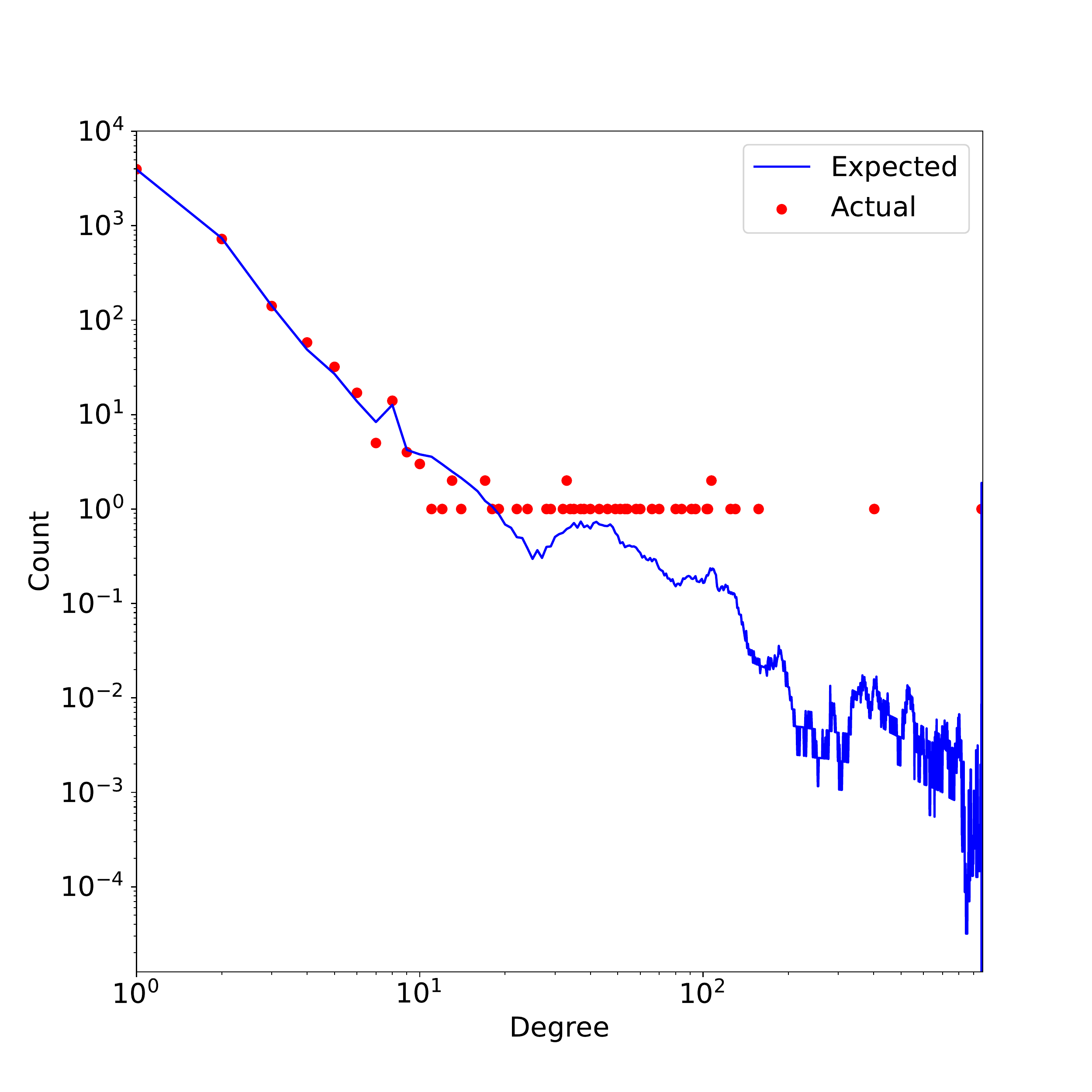}
  }
  \hfill
  \subfloat[][Bump Power-Law\label{F:degree_plb}]{
    \includegraphics[trim = 18 33 71 76, clip, width = .4\textwidth]{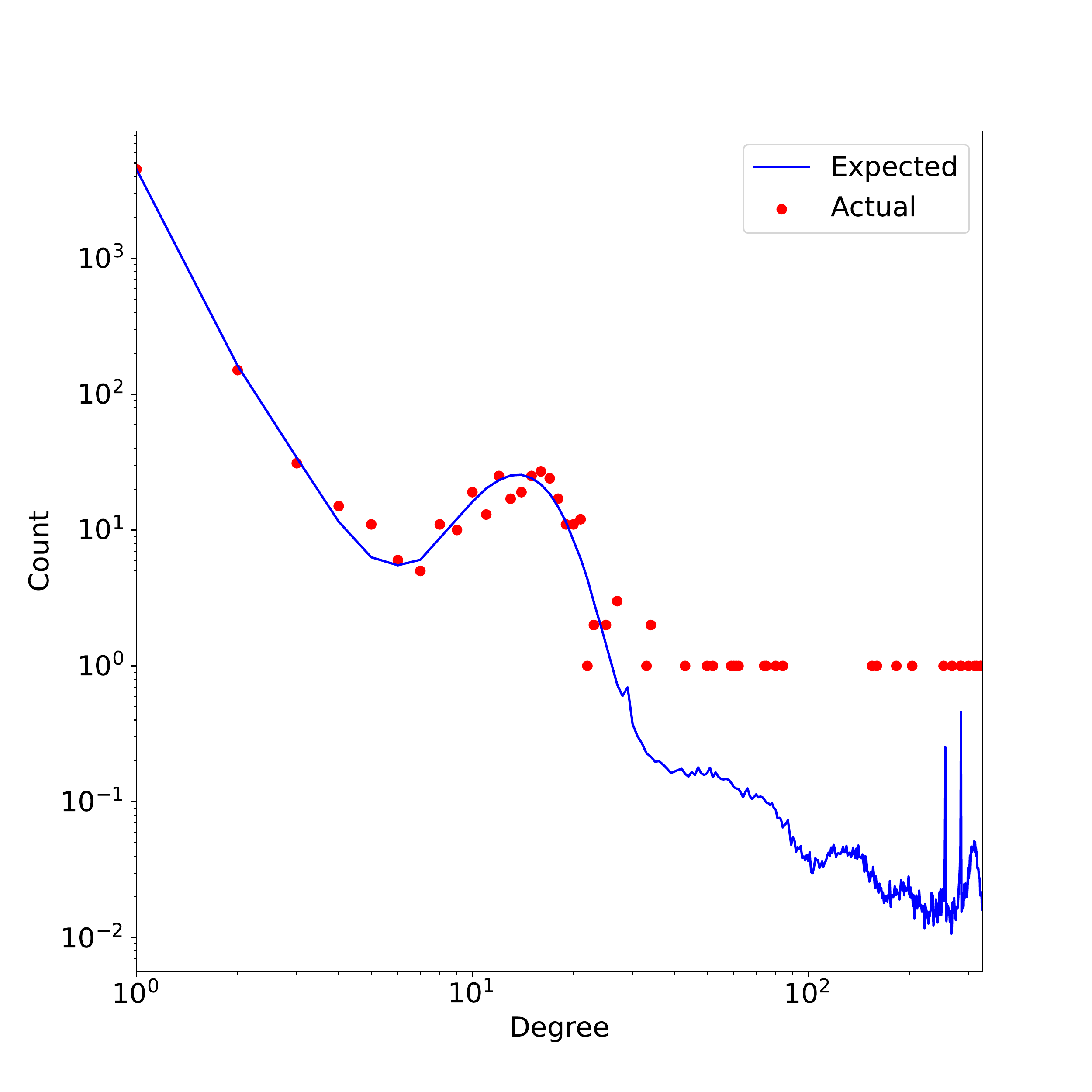}
  }
  \hfill\phantom{}

  \caption{$\log$-$\log$ Degree Distribution\label{F:degdist}}
  \end{figure}

We will also consider two different styles of anomalies involving between 10 and 50 edges.  The first anomaly involves three randomly chosen vertices adding some number of edges to the rest of the network uniformly at random.  We view this as behavior consistent with a probe or scan of the network structure.  For the second anomaly a random collection of vertices are chosen and a random spanning tree is added among those vertices.  We view this as behavior consistent with lateral movement scenario where an attacker is exploring the network by moving from machine to machine. They may backtrack and try different routes (thus a tree rather than just a path) as needed.

For each of these 420 scenarios (two different degree distributions, two different anomaly types, five different anomaly sizes, and 21 different values of $\alpha$) we produce 1,000 different pairs of graphs $(G,G')$ where $G$ is a random instance of the Chung-Lu model with the chosen degree distribution and $G'$ is formed from $G$ by performing one step of the HLM evolution with the chosen parameter $\alpha$, and then adding a random instance of the chosen anomaly of the chosen size. In this way, the anomaly occurs concurrently with the natural evolution of the network, as might be typical of real-world data.
For each of our 420 scenarios, and 1,000 pairs of graphs within the scenario, we compute the RH distance between $G$ and $G'$ to get a distribution of RH distances for the anomalous transition.

\subsection{Anomalous versus Nominal Relative Hausdorff Distance}
In this section we consider whether the anomalous transitions in the HLM model result in a different distribution of RH distances than a nominal transition.  To this end, for each degree distribution and choice of $\alpha$, we simulate 10,000 different HLM transitions to develop a baseline distribution of RH distances, see Figure \ref{F:RHbaseline}.  For each of the 420 anomalous scenarios we calculate the 2-sample Kolmogorov-Smirnov $p$-value~\cite{Young1977} between the previously calculated anomalous distribution and this baseline distribution.  For each of the 420 different anomaly scenarios the KS test significance value is less that $0.01$, indicating that we can reject the null hypothesis that the distribution of RH distances for an anomalous HLM transition is the same as the distribution for non-anomalous transitions.  In particular, this means that in a statistical sense the RH distance is able to pick up on anomalous evolution of the degree distribution, even when the anomaly only consists of 10 edges.\footnote{It is important to note that this is not always the case.  For example, in an experiment that is not reported for space limitations, we synthetically generated a degree distribution with a power-law exponent of 4 and average degree around 1.4, the anomalies resulted in a range of KS statistics including several scenarios which were not statistically distinguishable.  }  In the next subsection we will consider the effectiveness of the RH distance in detecting anomalous behavior directly, rather than statistically.

\begin{figure}
  \centering
  \hfill
  \subfloat[][Power-Law\label{F:RH_PLbaseline}]{
    \includegraphics[trim = 37 21 71 68, clip, width = .45\textwidth]{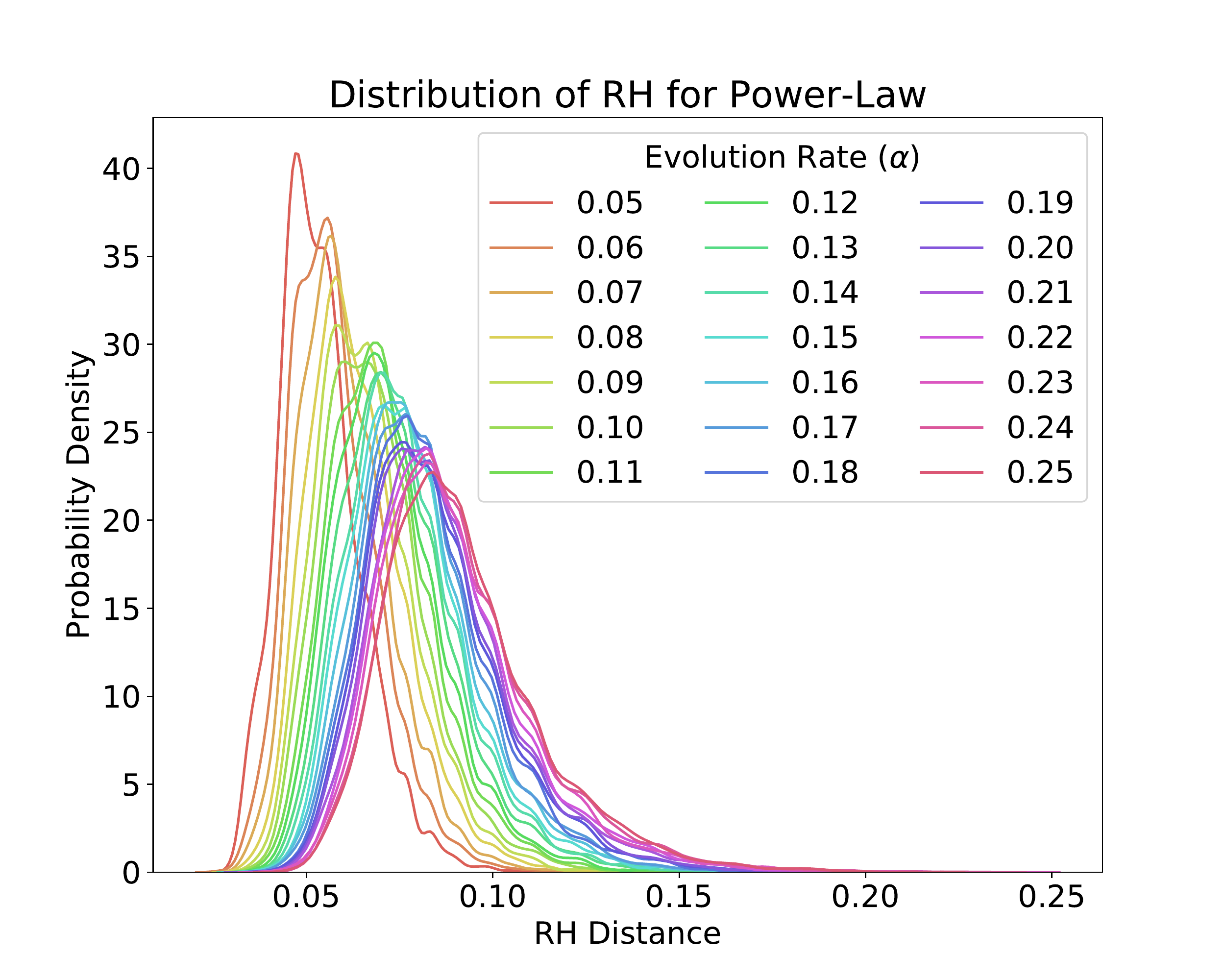}

  }
  \hfill
  \subfloat[][Bump Power-Law\label{F:RH_BumpPLbaseline}]{
    \includegraphics[trim = 37 21 71 68, clip, width = .45\textwidth]{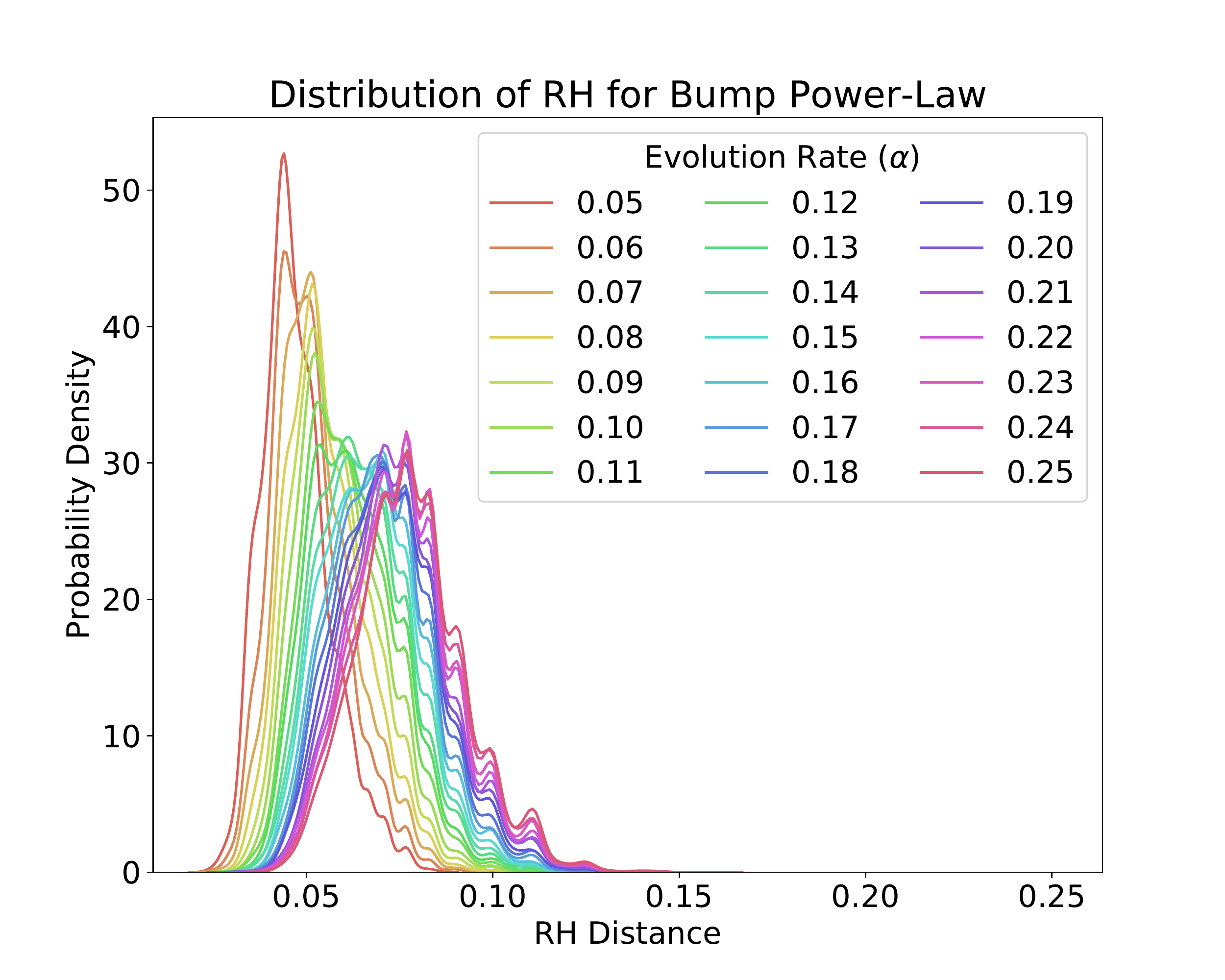}

  }
\hfill \phantom{}

  \caption{Distribution of RH distance under HLM evolution\label{F:RHbaseline}}
\end{figure}

\subsection{Anomaly Detection}
In this section we consider how RH distance could be used to detect anomalous behavior in a streaming environment and compare to the performance with a similarly lightweight measure (KS distance) as well as a ``ideal'' measure (graph edit distance).   To compare between these three methods in a non-parametric way (i.e. without introducing a ``anomaly threshold'') we introduce the  idea of an \emph{anomaly score} of an observation with respect to a theoretically or empirically observed baseline
distribution.  Specifically, let the random variable $Z$ have
theoretical or empirical cumulative distribution function $f_Z \colon
\R \rightarrow [0,1]$.  We will then say that a particular observation $z$ (not necessarily
distributed as $Z$) has an anomaly score relative to $f_Z$ of
$2\abs{f_Z(z) - \frac{1}{2}}$.  Note that this score takes on values
from $[0,1]$ with values closer to one being more ``anomalous.'' This score can be thought of as measuring
the deviation of the observation $z$
from the bulk of the distribution of $Z$.

Before turning to a direct comparison between anomaly scores for RH distance,  KS $p$-values\footnote{From this point on in this work, although we will be using the KS $p$-value we will be treating it simply as a distance measure rather than a statistical quantity.  In particular, we will make no assumptions about the meaning of large or small values of the $p$-value other than as a means of measuring the ``closeness'' between two degree distributions.}, and edit distance, we consider the performance of each of these anomaly scores in isolation via ROC-like curves, presented for a subset of our 20 scenarios (2 distributions, 2 anomaly types, 5 levels of each anomaly) in Figures \ref{F:ROC_PLPS10}, \ref{F:ROC_PLEX50}, \ref{F:ROC_BumpPLEX10}, and \ref{F:ROC_BumpPLPS50}.  Note that as the relative frequency of anomalous and non-anomalous behavior is unknown, these are not truly ROC curves but rather implicit plots $(x(t),y(t))$ where $t$ is some threshold value.  Specifically, $y(t)$ is the fraction of the anomalous transitions that have anomaly score at least $t$, i.e. ``true positives'', where $x(t)$ is the fraction of non-anomalous transitions that have anomaly score at least $t$, i.e. ``false positives''.  At this point it is worth pointing out that if $z$ is identically distributed with the random variable $Z$, then the anomaly score for $z$ is uniformly distributed over $[0,1]$.  As a consequence, we can explicitly define $x(t) = 1-t$.
To compute the ROC curve for one scenario we used the previously computed cumulative distribution function for the 10,000 non-anomalous transitions as $f_Z$.
Then, for each of the 1,000 anomalous transitions we use the RH distance as $z$ and compute the anomaly score for that value in the context of $f_Z$.

\begin{figure}
  \centering
  \hfill
  \subfloat[][Kolmogorov-Smirnov]{\includegraphics[trim = 30 35 70 85, clip, width = .3\textwidth]{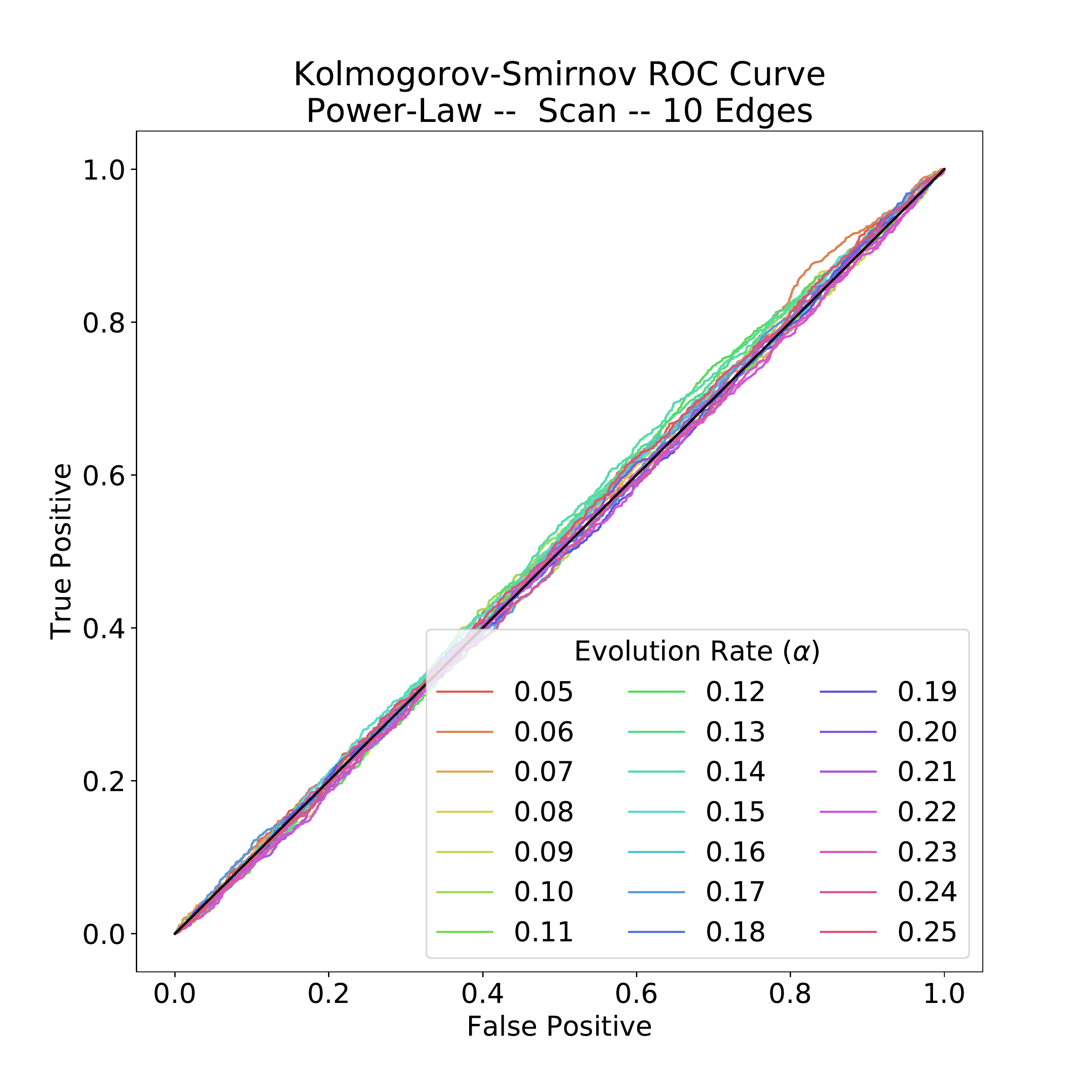}}
  \hfill
  \subfloat[][Relative Hausdorff]{\includegraphics[trim = 30 35 70 85, clip, width = .3\textwidth]{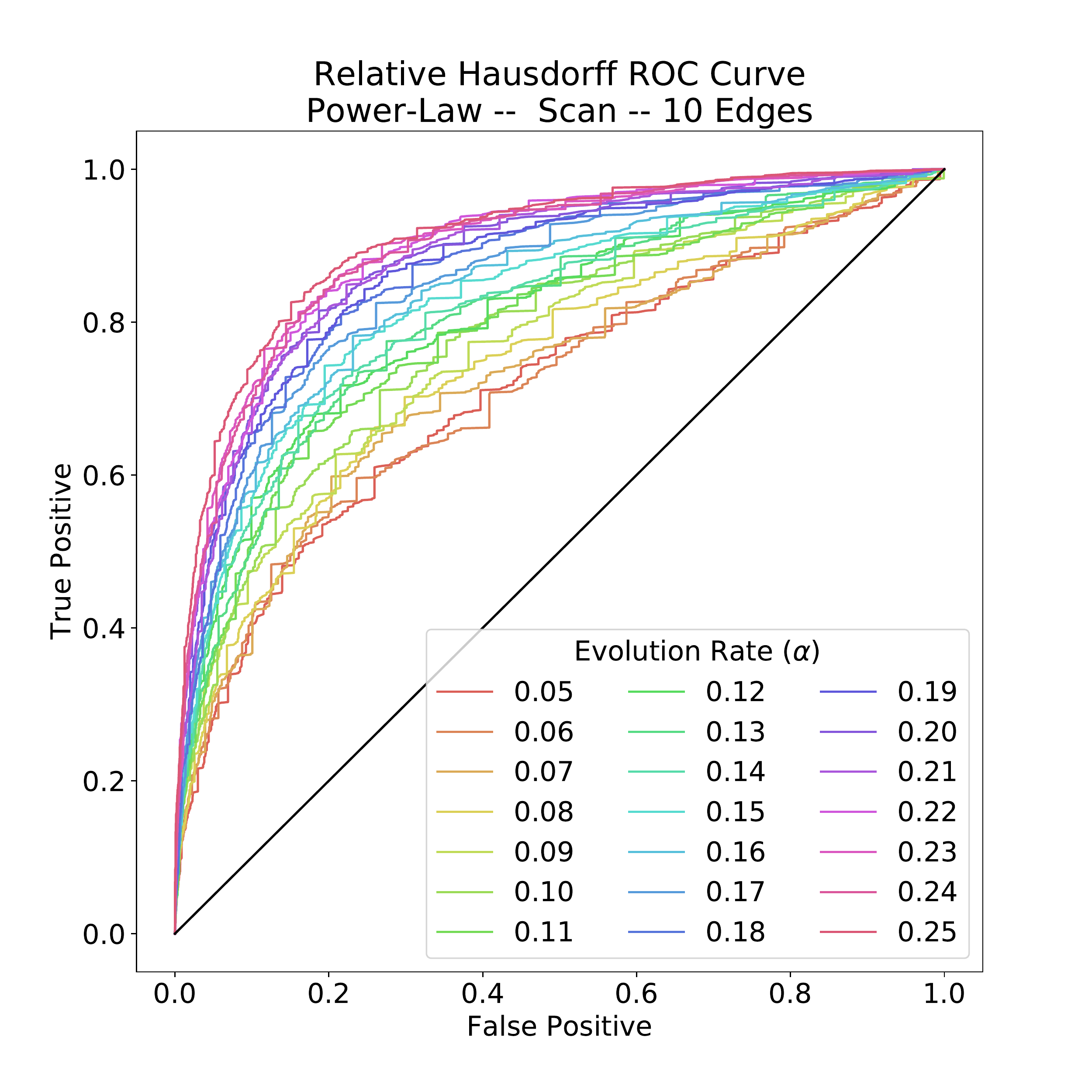}}
    \hfill
\subfloat[][Edit]{\includegraphics[trim = 30 35 70 85, clip, width = .3\textwidth]{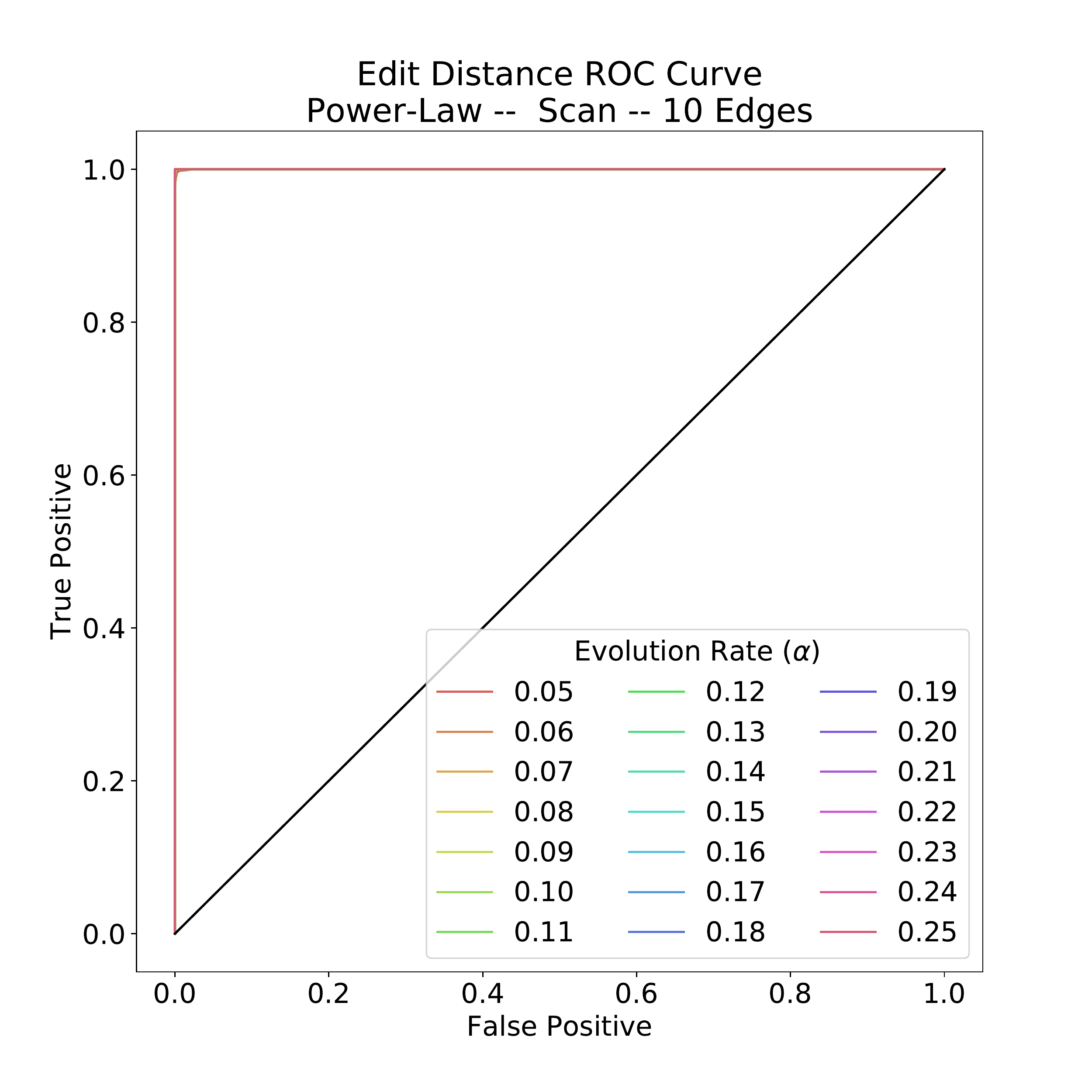}}
  \hfill \phantom{}
  \caption{Power-Law Distribution, Scan,  10 edges} \label{F:ROC_PLPS10}
\end{figure}

\begin{figure}
  \centering
  \hfill
  \subfloat[][Kolmogorov-Smirnov]{\includegraphics[trim = 30 35 70 85, clip, width = .3\textwidth]{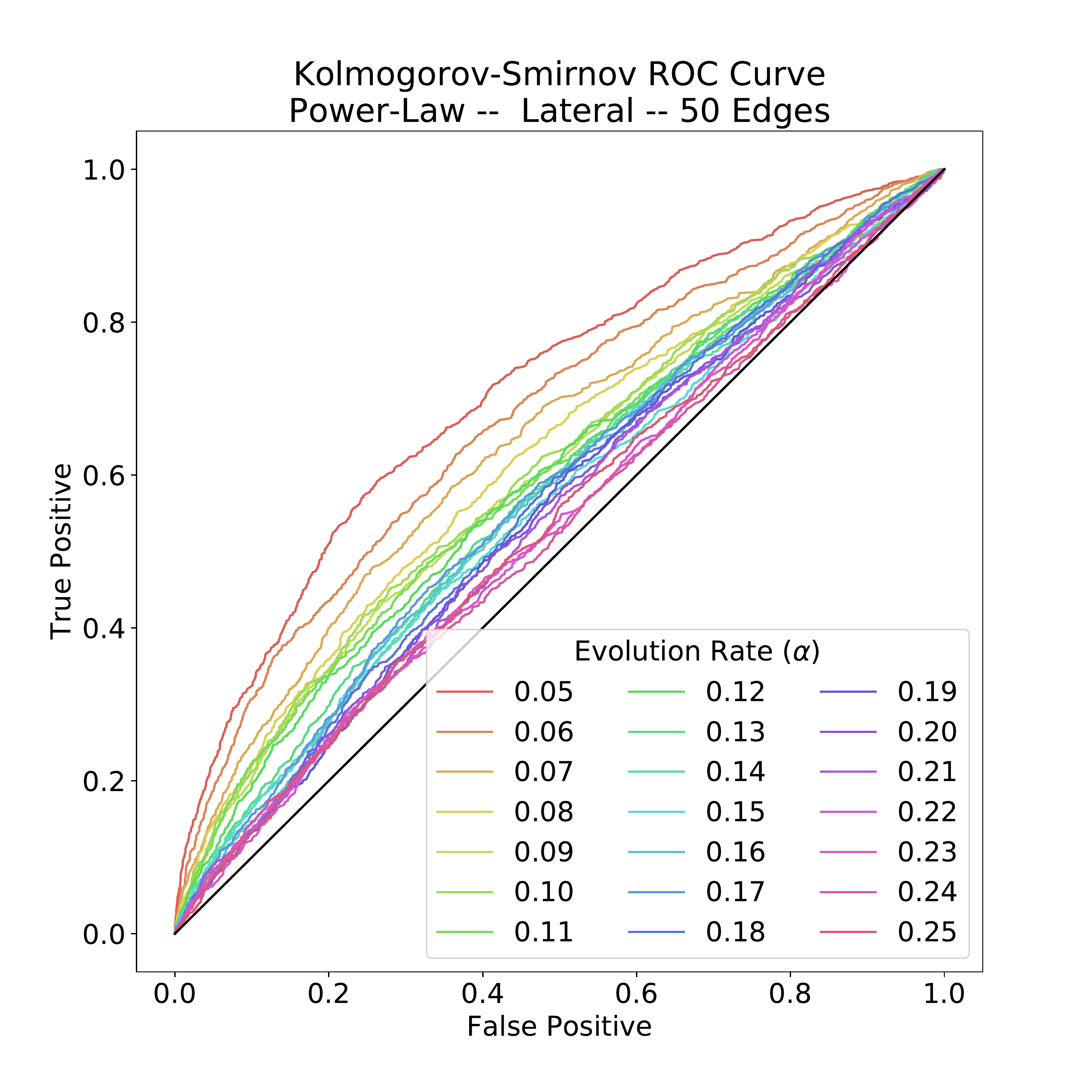}}
  \hfill
  \subfloat[][Relative Hausdorff]{\includegraphics[trim = 30 35 70 85, clip, width = .3\textwidth]{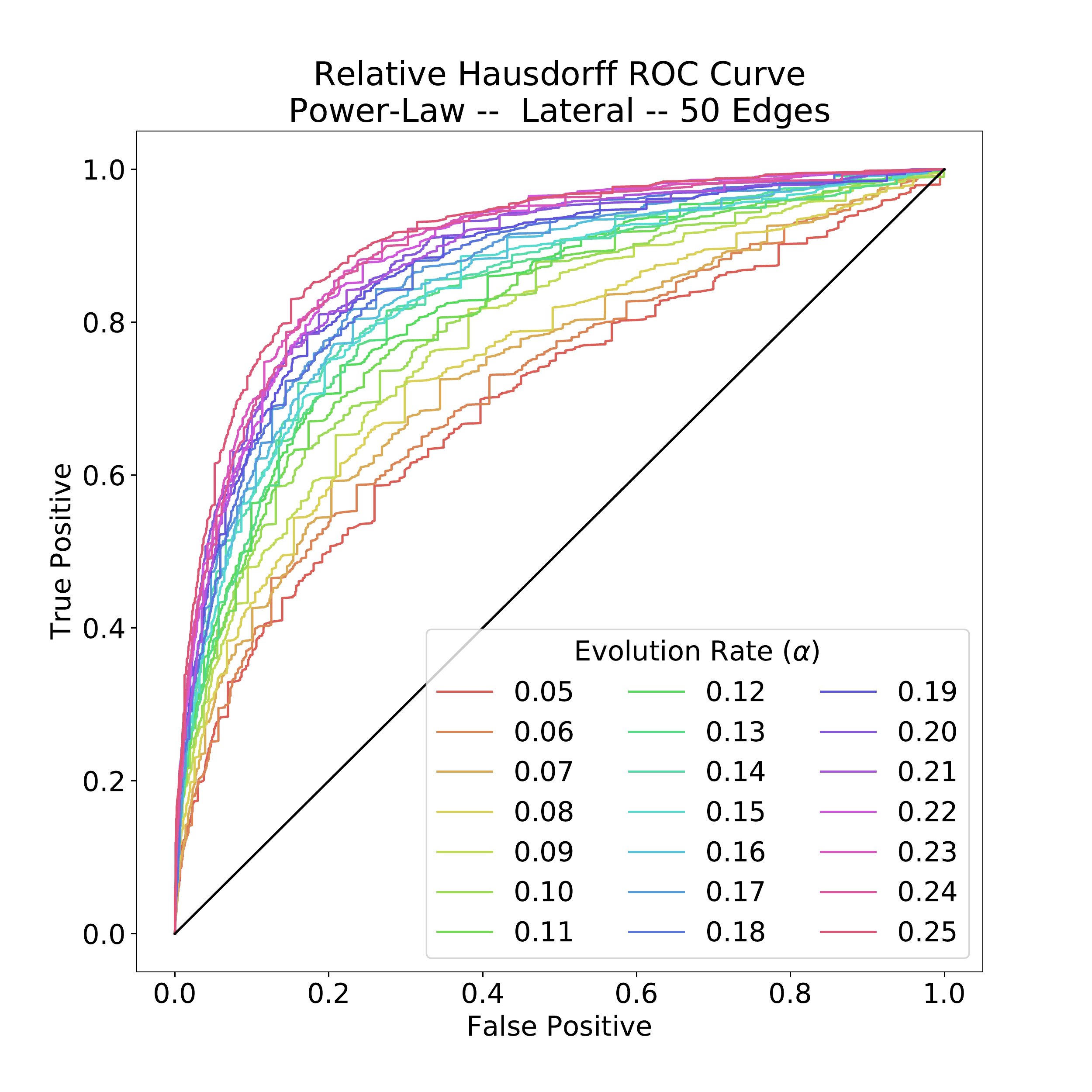}}
    \hfill
\subfloat[][Edit]{\includegraphics[trim = 30 35 70 85, clip, width = .3\textwidth]{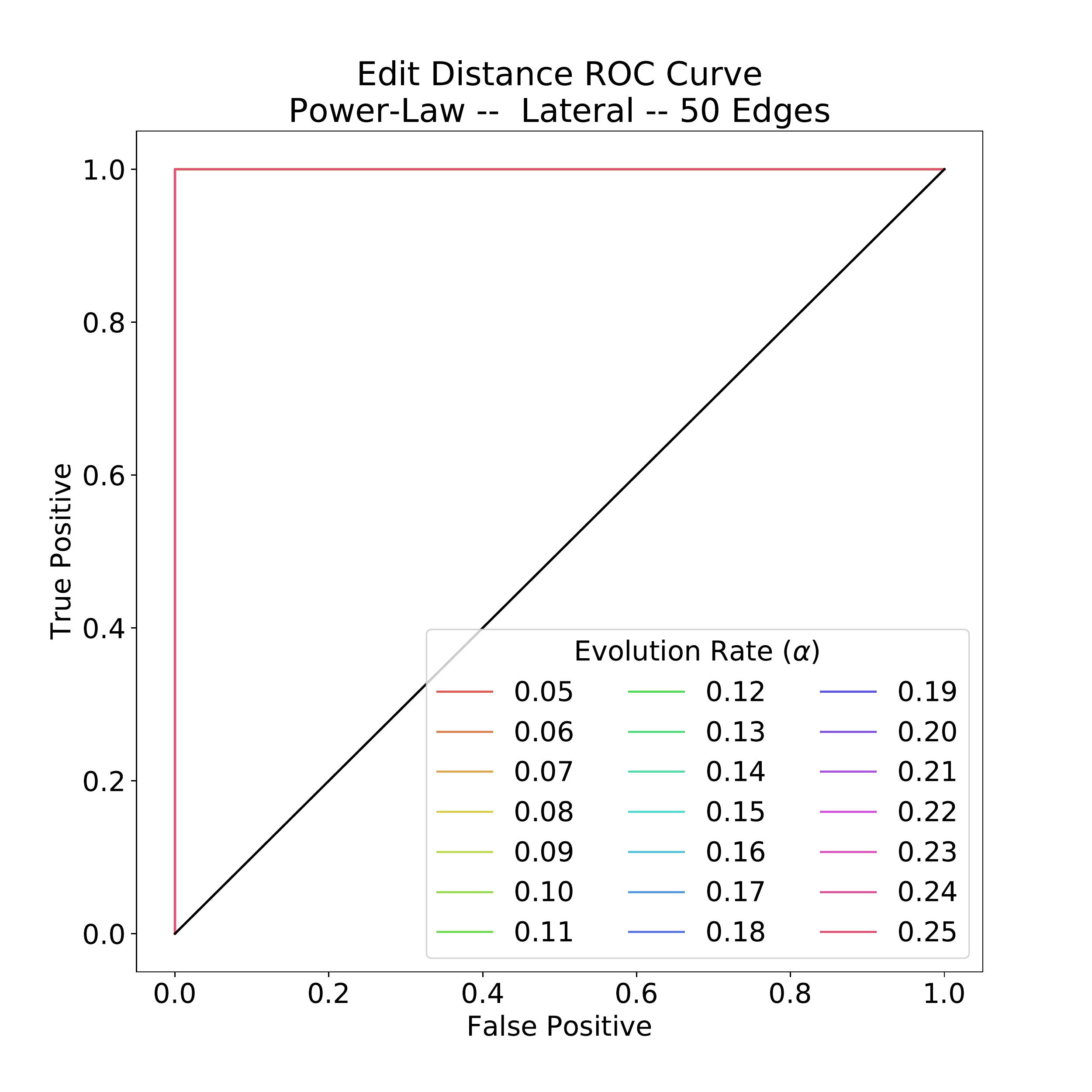}}
  \hfill \phantom{}
  \caption{Power-Law, Lateral Movement, 50 edges} \label{F:ROC_PLEX50}
\end{figure}

\begin{figure}
  \centering
  \hfill
  \subfloat[][Kolmogorov-Smirnov]{\includegraphics[trim = 30 35 70 85, clip, width = .3\textwidth]{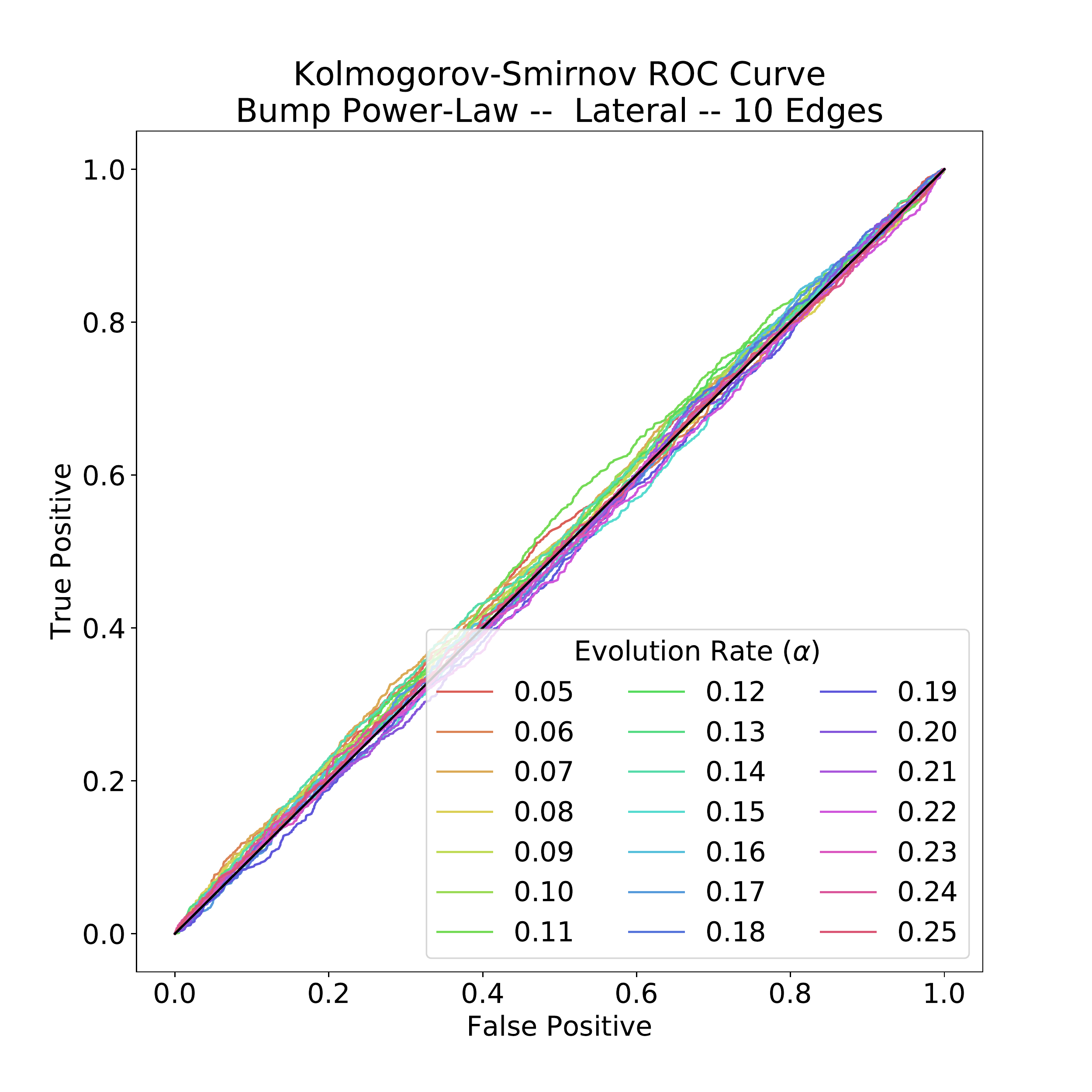}}
  \hfill
  \subfloat[][Relative Hausdorff]{\includegraphics[trim = 30 35 70 85, clip, width = .3\textwidth]{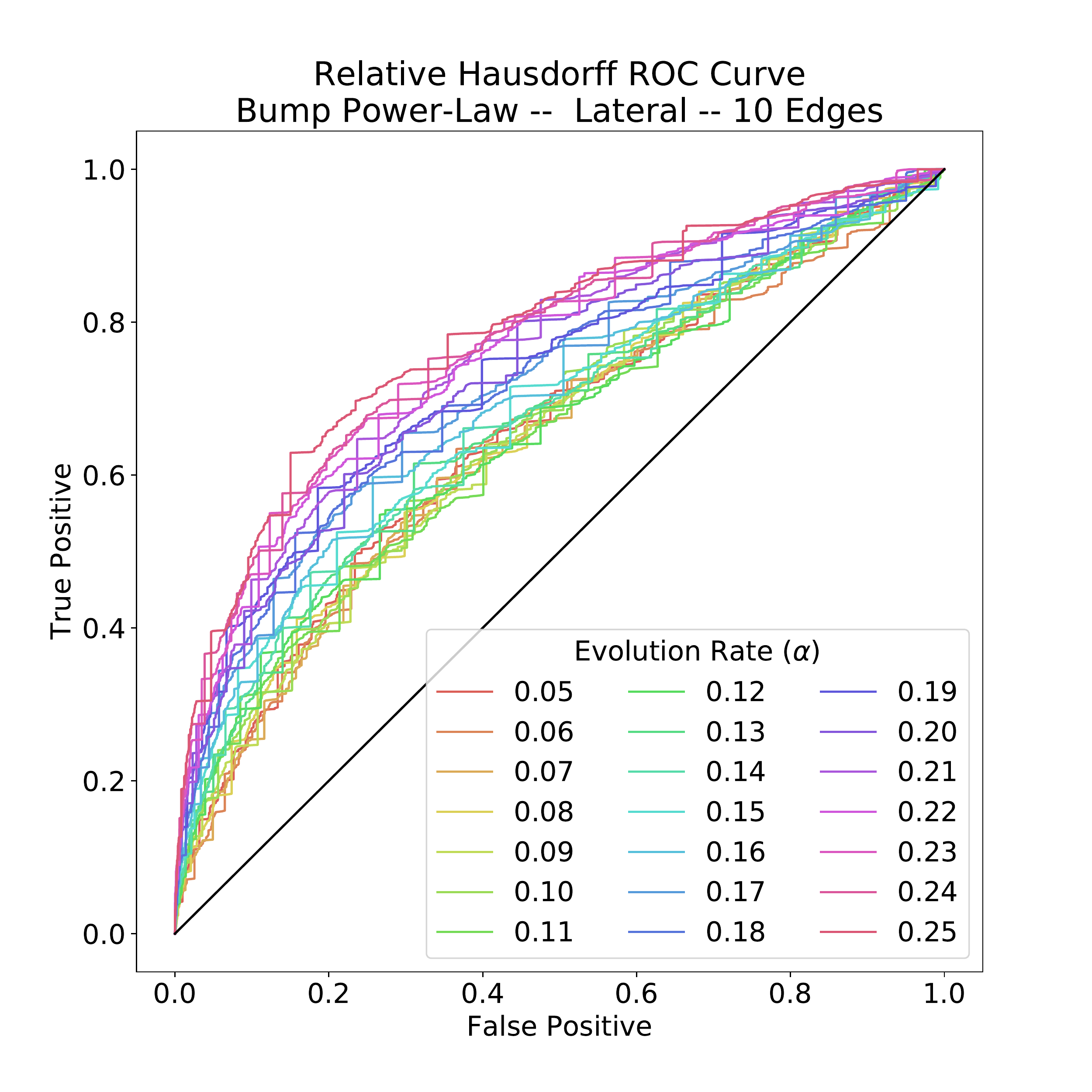}}
    \hfill
\subfloat[][Edit]{\includegraphics[trim = 25 40 60 85, clip, width = .3\textwidth]{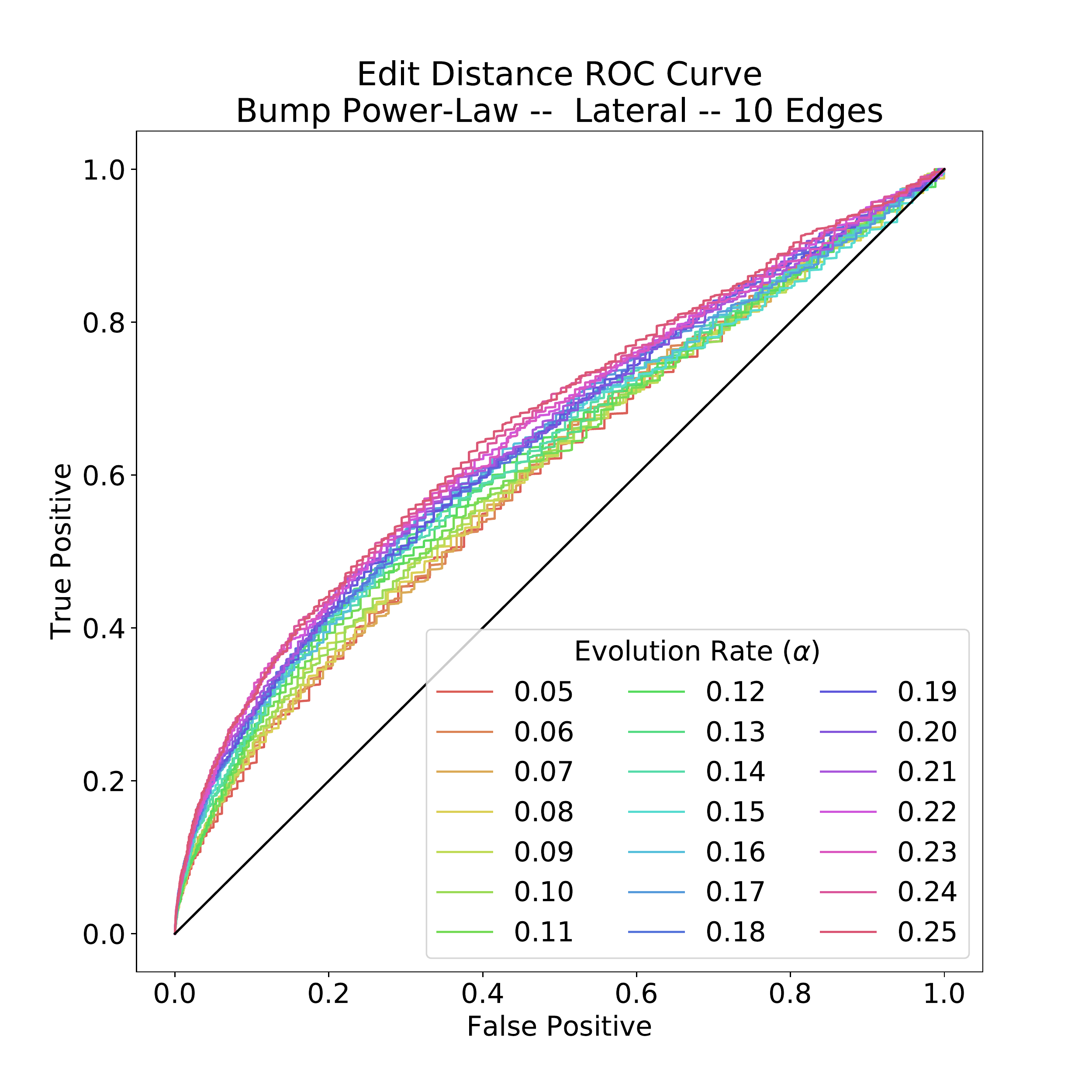}}
  \hfill \phantom{}
  \caption{Bump Power-Law, Lateral Movement, 10 edges} \label{F:ROC_BumpPLEX10}
\end{figure}

\begin{figure}
  \centering
  \hfill
  \subfloat[][Kolmogorov-Smirnov]{\includegraphics[trim = 30 35 70 85, clip, width = .3\textwidth]{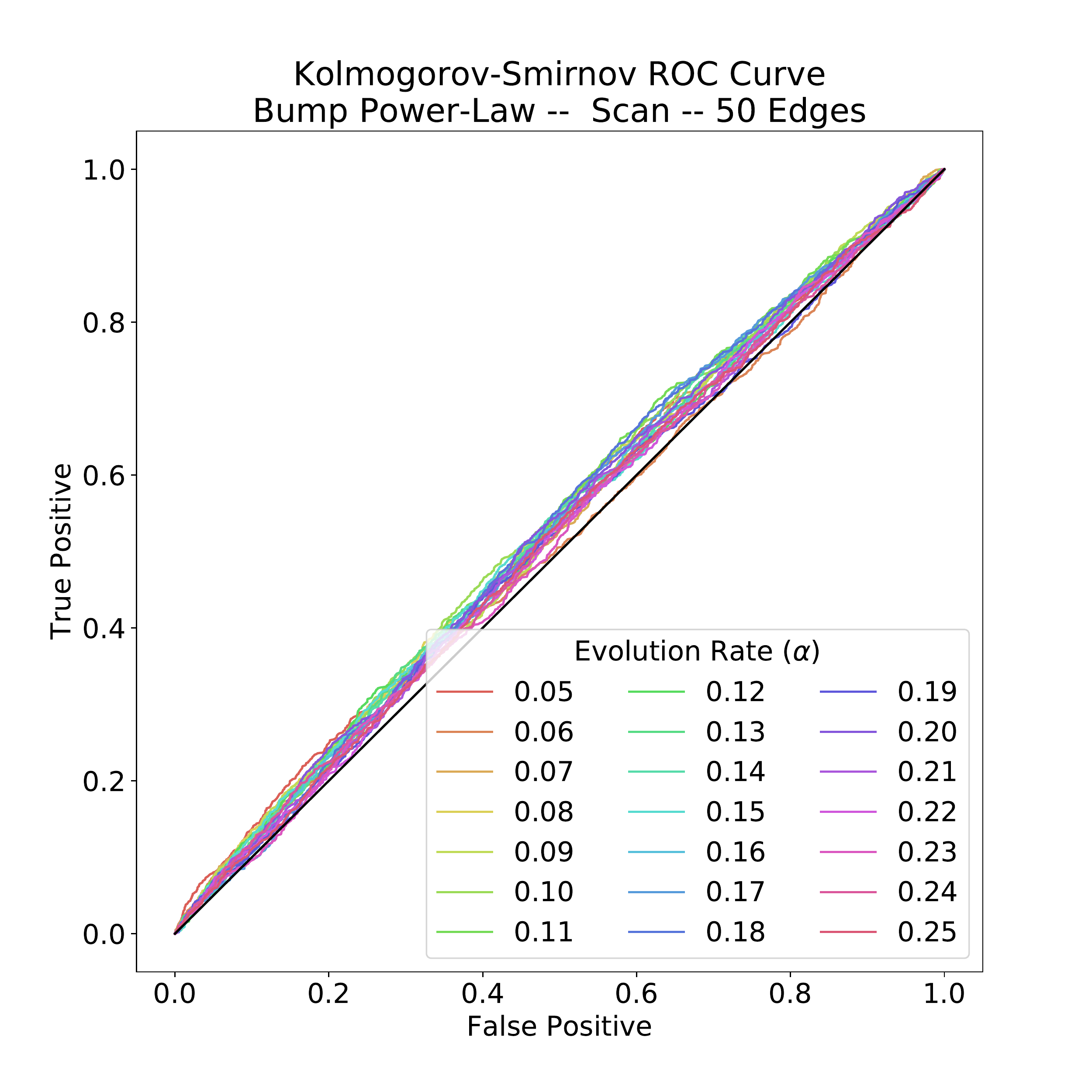}}
  \hfill
  \subfloat[][Relative Hausdorff]{\includegraphics[trim = 30 35 70 85, clip, width = .3\textwidth]{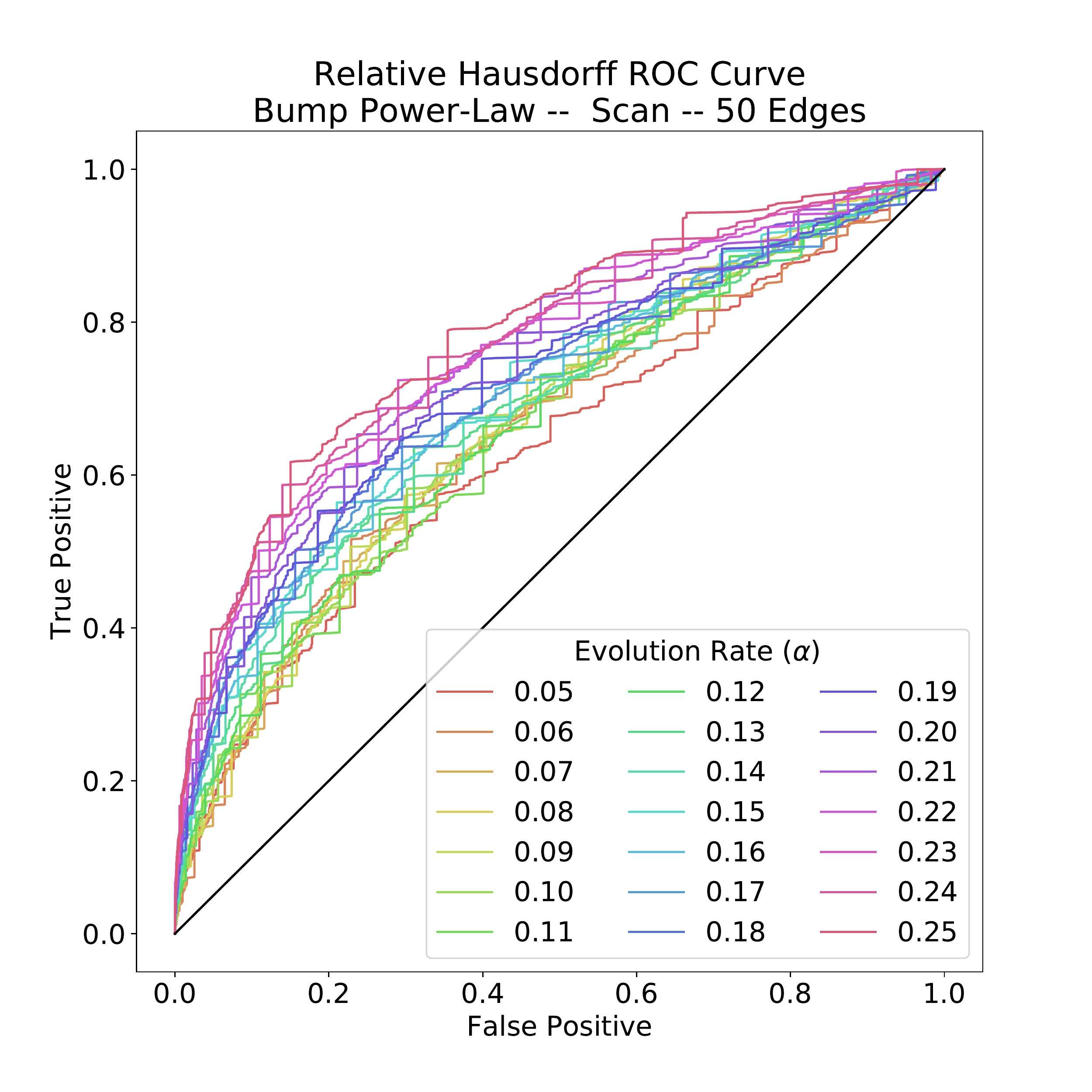}}
    \hfill
\subfloat[][Edit]{\includegraphics[trim = 30 35 70 85, clip, width = .3\textwidth]{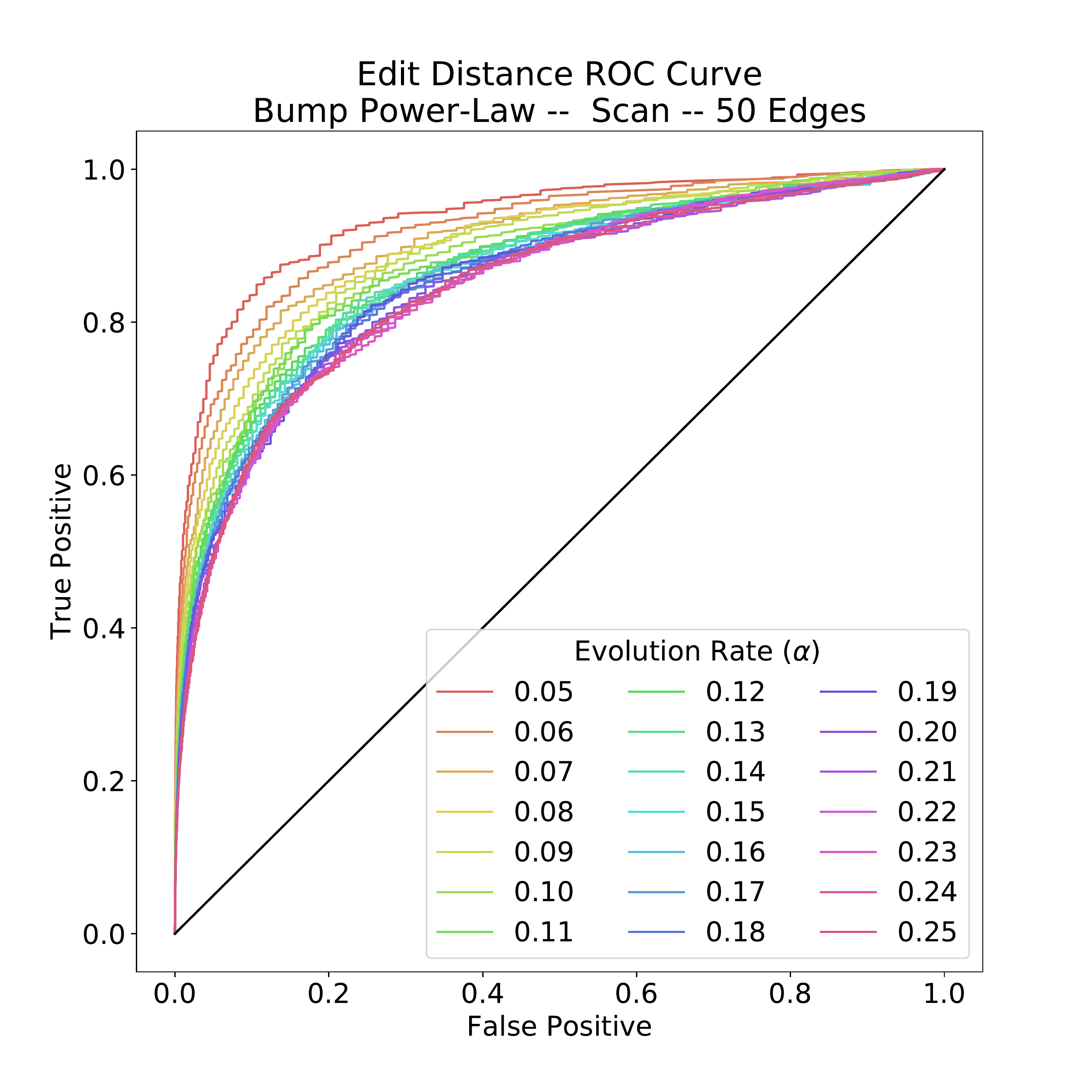}}
  \hfill \phantom{}
  \caption{Bump Power-Law, Scan, 50 edges} \label{F:ROC_BumpPLPS50}
\end{figure}

Overall we can see that, for detecting anomalies, edit distance would be preferred to RH distance, which would in turn be preferred to the KS statistic.  However, for the bump power-law, under the lateral movement anomaly with 10 edges, we see that the RH distance outperforms the edit distance, see Figure \ref{F:ROC_BumpPLEX10}.

\subsubsection{Kolmogorov-Smirnov (KS)}
For this section we will compare the anomaly score of the RH distance with the anomaly score of the KS $p$-value (significance value) between successive degree distributions both for the 420 anomalous scenarios and the 40 baseline distributions.  We note that since the degree distributions are discrete valued, the application of the KS test for hypothesis testing is not necessarily appropriate, however as we are interested in statistical behavior of the significance test and KS is widely used in the network analysis literature (see, for instance, \cite{aliakbary2014quantification, broido2018scale, Simpson2015}) we will ignore these technical issues.

Figure \ref{F:KSH2H} gives the relative performance of the KS and RH anomaly scores across all 420 anomaly scenarios. The $y$ value counts how many of the 1,000 cases RH outperforms KS. We can see that the RH distance outperforms KS the most for the scenario where there is a 50-edge scan anomaly on the  power-law distribution with an evolution rate of $0.23$.  In this case, the RH anomaly score is larger in 924 of the 1,000 different trials.  We see that overall, excepting cases with a low-evolution rate and larger, lateral movement anomalies, RH distance is clearly superior, especially for the power-law degree distribution.  It is worth noting that relative performance of the KS statistic improves when considering the lateral movement anomaly rather than the scan anomaly.  Since the degree change caused by lateral movement is spread across many vertices (as opposed to scan where the primary change is spread across only three vertices), this result can be explained by the well known sensitivity of the KS test to variation away from the a tails of the distribution~\cite{Simpson2015}. 

\begin{figure}
  \centering
\hfill
 \subfloat[][RH versus KS\label{F:KSH2H}]{
   \includegraphics[trim = 47 29 69 85 , clip, width = .4\textwidth]{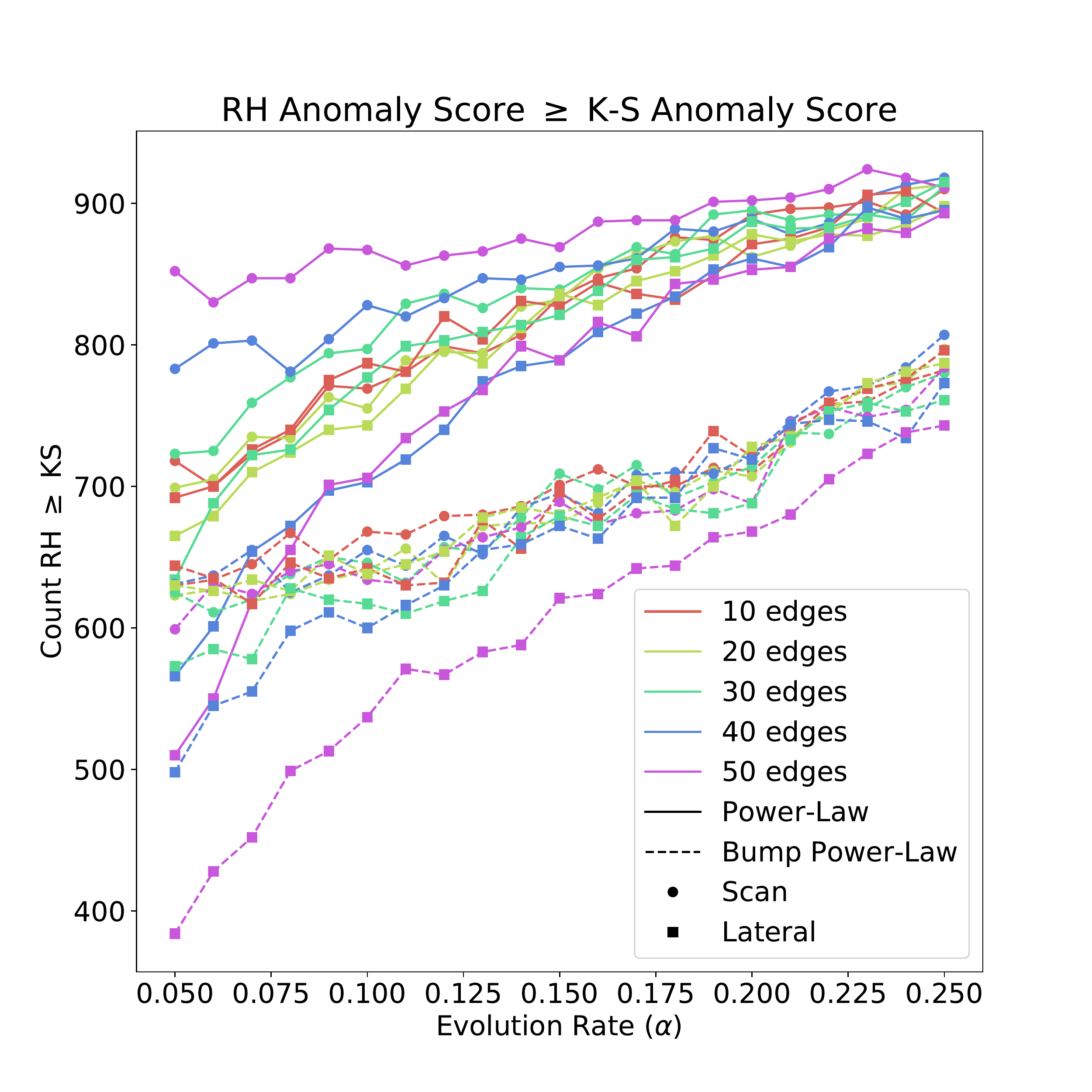}}
  \hfill
  \subfloat[][RH versus Edit\label{F:EditH2H}]{
    \includegraphics[trim = 47 29 69 85 , clip, width = .4\textwidth]{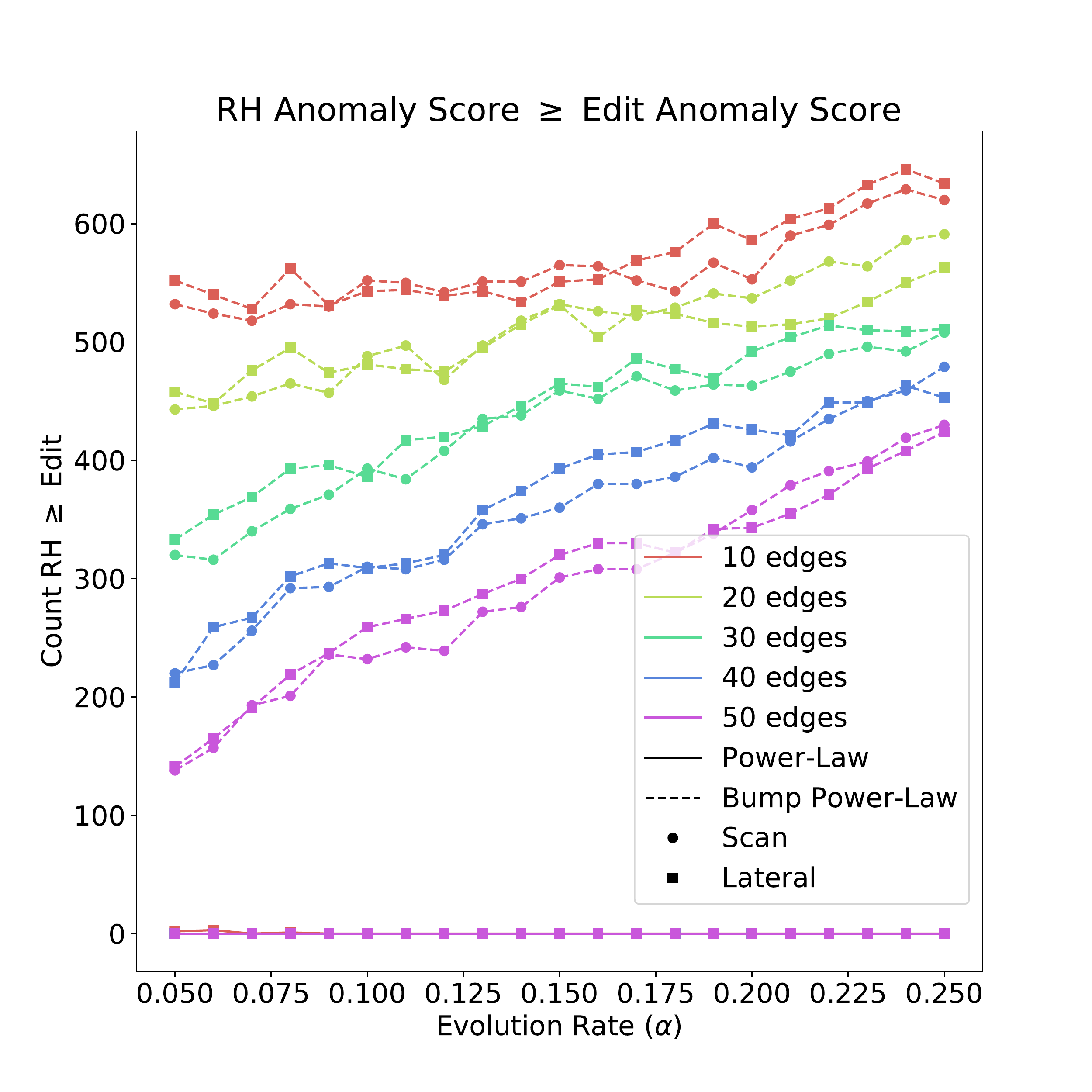}}
  \hfill\phantom{}
  \caption{Relative Performance of RH and KS Anomaly Scores}\label{F:KSH2H}

\end{figure}

We note that a direct binary comparison between the two measures may not tell the whole story of their relative performance.  For instance, in an extreme case, one can imagine one of the two measures taking on a fixed large value (indicating an anomaly) while the other takes on both small values and, more frequently, a value that is slightly larger than the other anomaly score. We separate the anomalous pairs of graphs into two sets according to whether their RH or KS anomaly score is higher. 
Then, for each of these two classes we report in Figure \ref{F:KS_H2Hrange} (across all 420 anomaly scenarios) the mean difference between the scores with error bars representing one standard deviation of range around this mean.
We note that the RH anomaly score typically exceeds the KS anomaly score by about 0.4, while the KS anomaly score typically exceeds the RH anomaly score by between $0.2$ and $0.3$.  Further, the standard deviation across all cases the average gap between the RH and KS anomaly score is fairly consistently  in the range $[0.2,0.3]$ essentially independent of all parameters.  Together, the data in Figures \ref{F:KSH2H} and \ref{F:KS_H2Hrange}, indicates that the RH distance is significantly more sensitive than KS distance to the broad range of anomalies we have investigated. 
\begin{figure}
  \centering
  \hfill
  \subfloat[][Power-Law \\ Scan]{
    \includegraphics[trim = 15 28 69 85, clip, width = .4\textwidth]{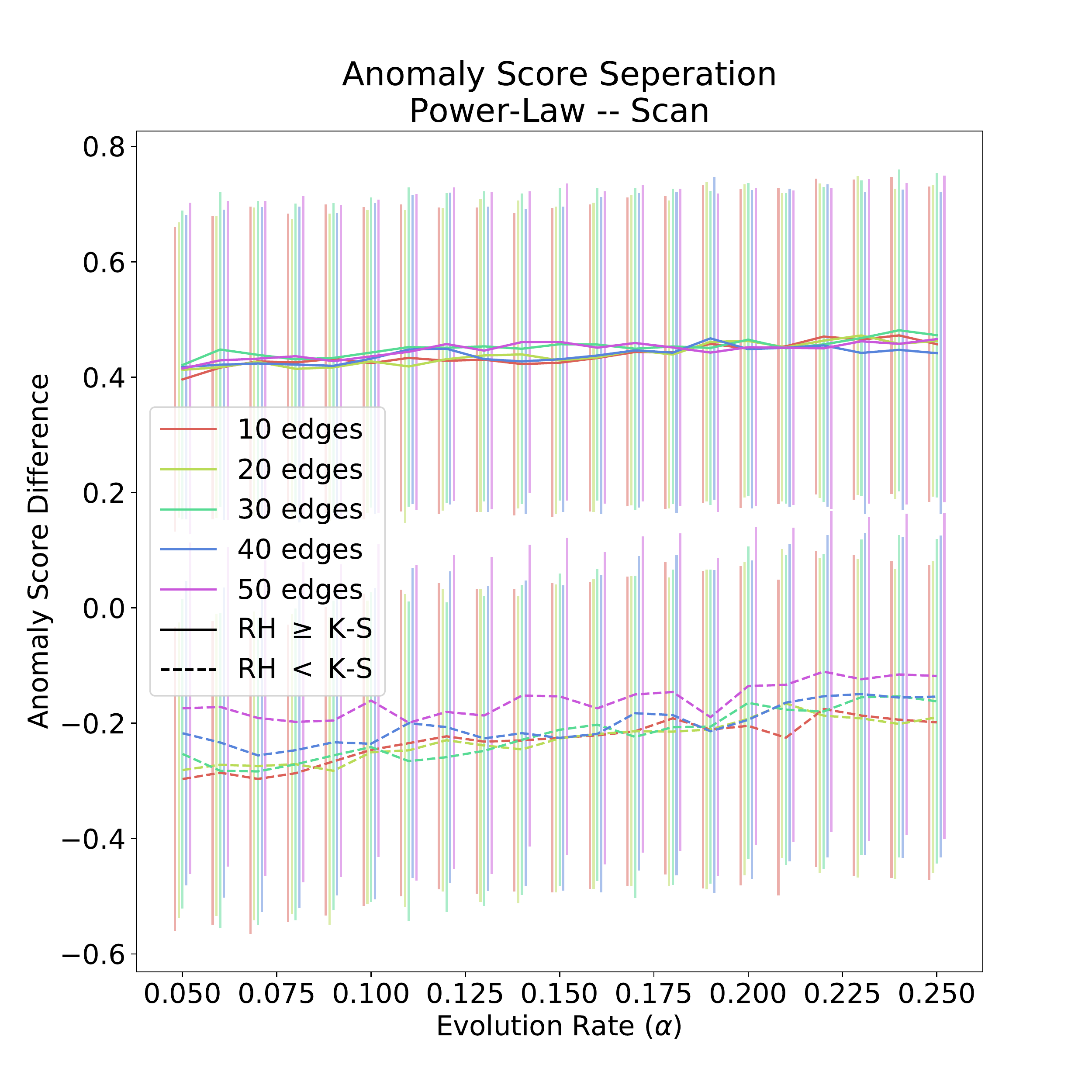}

  }
  \hfill
  \subfloat[][Bump Power-Law \\ Scan]{
    \includegraphics[trim = 15 28 69 85, clip, width = .4\textwidth]{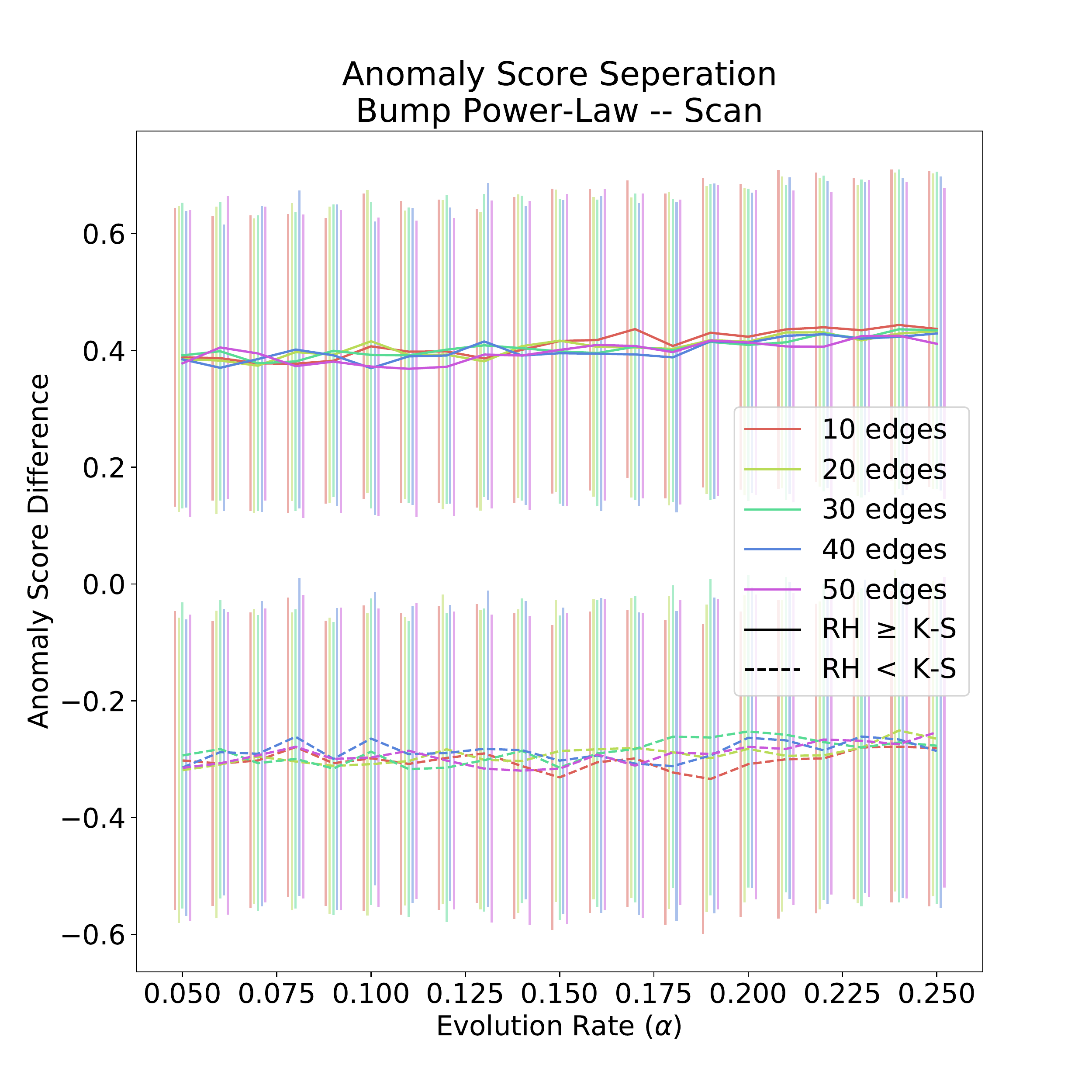}

  }
  \hfill \phantom{} \\
  \hfill
  \subfloat[][Power-Law \\ Lateral movement]{
    \includegraphics[trim = 15 28 69 85, clip, width = .4\textwidth]{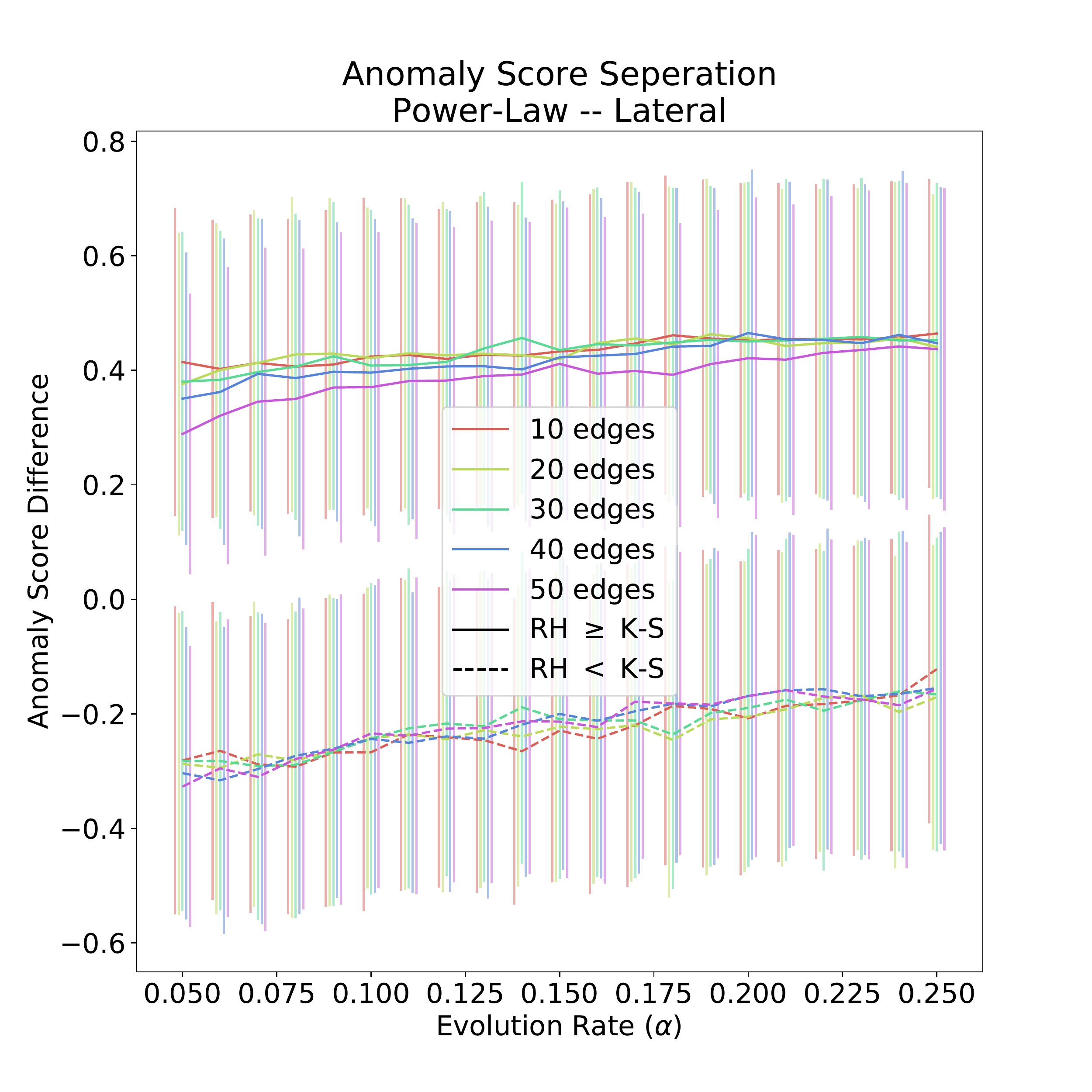}

  }
  \hfill
  \subfloat[][Bump Power-Law \\ Lateral movement]{
    \includegraphics[trim = 15 28 69 85, clip, width = .4\textwidth]{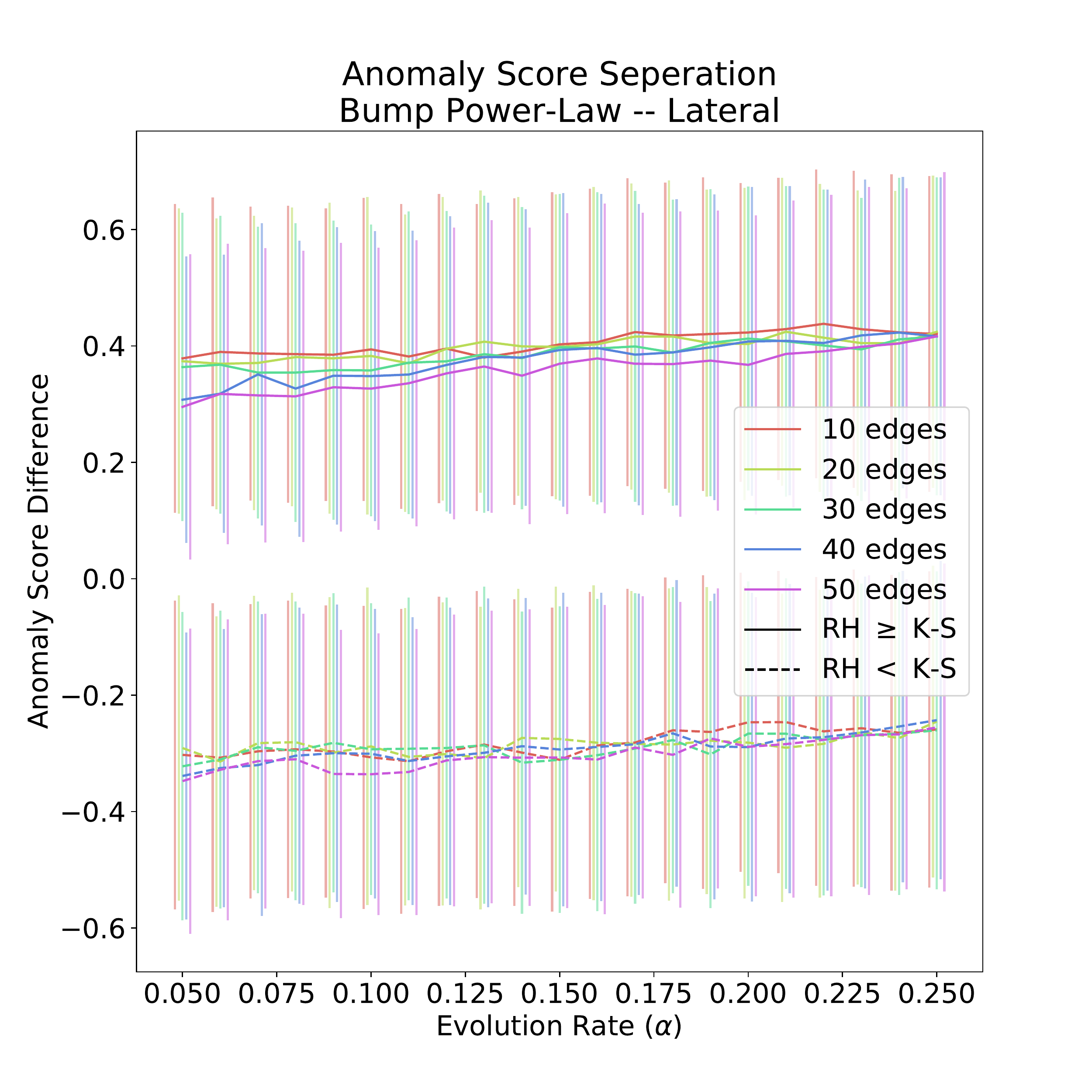}

  }
  \hfill \phantom{}
  \caption{Difference between KS anomaly scores and RH anomaly scores}\label{F:KS_H2Hrange}
  \end{figure}

\subsubsection{Graph Edit Distance}
In this section we compare the sensitivity of RH distance to a  ``perfect information'' aggregate measure, in particular graph edit distance.  Recall that the edit distance
between two graphs $F$ and $G$ is the minimum ``weight'' of a sequence
of edge/vertex additions/deletions needed to transform $F$ to $G$.  In
general this quantity is \NPcomp\ to compute~(see \cite{zeng2009comparing}) and likely
impractical to even approximate~\cite{Lin1994}.  This complexity is driven by
the difficulty in finding the optimal alignment between the vertices
of $F$ and $G$ which maximizes the edge overlap between $F$ and $G$.
For the HLM model, this problem is mitigated by the natural alignment
between the graphs generated at consecutive time steps.  Thus, for
purposes of this section we will approximate the graph edit distance
as the number of edges that ``flip'' during each evolution of the HLM
model.

The following lemma allows us to significantly simplify the calculation of the anomaly score for edit distance by approximating the baseline distribution with the large $n$ limit.
\begin{lemma}
Let $G$ be an random graph distributed according to $\mathcal{G}(P)$
and let $G'$ be the graph formed by one iteration of the
Hagberg-Lemons-Mishra evolution with evolution parameter $\alpha$ and
probability matrix $P'$.  Let $X$ be the random variable
that counts the number of edges that differ
between $G$ and $G'$.  If $\Var{X}\rightarrow \infty$, then $X$ is asymptotically normally distributed.
\end{lemma}

\begin{proof}
  Let $X_{ij}$ be the indicator function for the random variable that
  edge $\set{i,j}$ is present in precisely one of $G'$ and $G$ and
  observe that $X = \sum_{i < j} X_{ij}$.  We
recall that by the Lyapunov Central Limit Theorem~\cite[p. 362]{billingsley2008probability}, we have that
 \[ \frac{X - \expect{X}}{\sqrt{\Var{X}}} \overset{\mathcal{D}}{\longrightarrow} \mathcal{N}(0,1)\]
   if there is some $\delta > 0$ such that
  \[\lim_{n \rightarrow \infty} \frac{1}{\Var{X}^{\nicefrac{2+\delta}{2}}}
    \sum_{i < j} \expect{\abs{X_{ij} - \expect{X_{ij}}}^{2+\delta}} =
    0.\]
  Fixing $\delta = 1$, we note that
  \begin{align*}
\sum_{i < j} \expect{\abs{X_{ij} - \expect{X_{ij}}}^{3}} &= \sum_{i <
                                                          j}
                                                          \expect{X_{ij}}
                                                          \paren{1-\expect{X_{ij}}}^3
                                                          +
                                                          \paren{1-\expect{X_{ij}}}\expect{X_{ij}}^3
    \\
    &= \sum_{i < j}
      \expect{X_{ij}}\paren{1-\expect{X_{ij}}}\paren{\expect{X_{ij}}^2
      + \paren{1-\expect{X_{ij}}}^2}  \\
    &\leq \sum_{i < j}
      \expect{X_{ij}}\paren{1-\expect{X_{ij}}} \\
    &= \Var{X}.
  \end{align*}
  Thus, if $\Var{X} \rightarrow \infty$, then
  \[\lim_{n \rightarrow \infty} \frac{1}{\Var{X}^{\nicefrac{3}{2}}}
    \sum_{i < j} \expect{\abs{X_{ij} - \expect{X_{ij}}}^{3}} =
    0\] and $X$ is normally distributed.
\end{proof}

It is worth mentioning that this, in principle, allows for an explicit
formula for the distribution of the anomaly score for edit distance
in a wide range of baseline and anomalous behaviors, namely\[ \prob{S \leq s} = \Phi\paren{\frac{\mu - \mu_A}{\sigma_A} +
    \frac{\sigma_A}{\sigma} \Phi^{-1}\paren{\frac{s+1}{2}}} - \Phi\paren{\frac{\mu - \mu_A}{\sigma_A} -
    \frac{\sigma_A}{\sigma} \Phi^{-1}\paren{\frac{s+1}{2}}},\]
where $(\mu,\sigma)$ and $(\mu_A,\sigma_A)$ are the mean and distribution of the baseline and anomalous evolutions, respectively, and $\Phi$ is the cumulative distribution function of the standard normal distribution.  However, given the correlated nature of the anomalies the calculation of $\sigma_A$ is tedious, so we will empirically estimate this distribution.

Figure \ref{F:EditH2H} again presents the relative performance of the anomaly scores, this time for edit distance and RH distance, for all 420 anomaly trials. Again the $y$ value counts how many of the 1,000 cases RH outperforms edit distance.  We note that in the best case (bump power-law degree distribution, lateral movement anomaly, 10 edges, $\alpha = 0.24$), the RH distance anomaly score is larger than the edit distance anomaly scores 646 times.  However, for the power-law case the RH anomaly scores essentially never outperform the edit distance anomaly scores. This failure is mitigated by the fact that, as mentioned earlier, in many cases the edit distance is computationally infeasible, while the RH distance requires minimal computational overhead. It is also worth mentioning that we can see a clear degradation of performance for edit distance as the size of the anomaly decreases and the evolution rate increases.  This phenomenon can be explained by observing that the anomaly score for edit distance is driven by a $z$-score of the anomaly, which is linearly correlated with the anomaly size and inversely correlated with the standard deviation of the baseline distribution.  Additionally, the variance baseline distribution of edit distance is linear related to the evolution rate, resulting in significantly decreased sensitivity at high evolution rates.

We further compare the relative behavior of the edit distance anomaly scores and the RH distance anomaly scores, in the same way as we did for KS above, by considering the average difference between the anomaly scores in the cases where the RH anomaly score is larger (positive values) and in the case the edit distance anomaly score is larger (negative values).  As the RH distance anomaly score is essentially never larger than the edit distance anomaly score for the power-law distribution, we restrict our attention here to the bump power-law distribution.  In Figure \ref{F:edit_H2Hrange}, we again report the relative magnitude of the differences with the error bars representing an interval one standard deviation away from the mean.  Again we can see a clear stratification of the behavior with the RH anomaly scores performing better as the size of the anomaly decreases.  We also note the mild improvement in the performance of RH distance as the evolution rate increases, likely reflecting the decreased sensitivity of edit distance (due to larger variance).

Interestingly, the standard deviation is essentially constant over all choices of degree distribution, anomaly type, anomaly size, and evolution rate and is also roughly equal to the standard deviations shown in Figure \ref{F:KS_H2Hrange}.  Furthermore, the magnitude of the standard deviation is close to the minimal possible standard deviation given by the generalization of Bhatia-Davis inequality for the variance of a bounded random variable~\cite{agarwal2005survey}.  As the extremal distribution is given by point masses at the end points of the distribution, this indicates that there are three essentially distinct outcomes: the RH distance anomaly score is significantly larger than the edit distance anomaly score, the RH and edit distance anomaly scores are essentially the same, and the edit distance anomaly score is significantly larger than the RH distance anomaly score.  Furthermore, this holds regardless of the size and nature of anomaly or evolution rate and also holds when replacing edit distance with KS distance (for both degree distributions).  

\begin{figure}
  \centering
  \hfill
  \subfloat[][Bump Power-Law \\ Scan]{
    \includegraphics[trim = 15 28 69 85, clip, width = .4\textwidth]{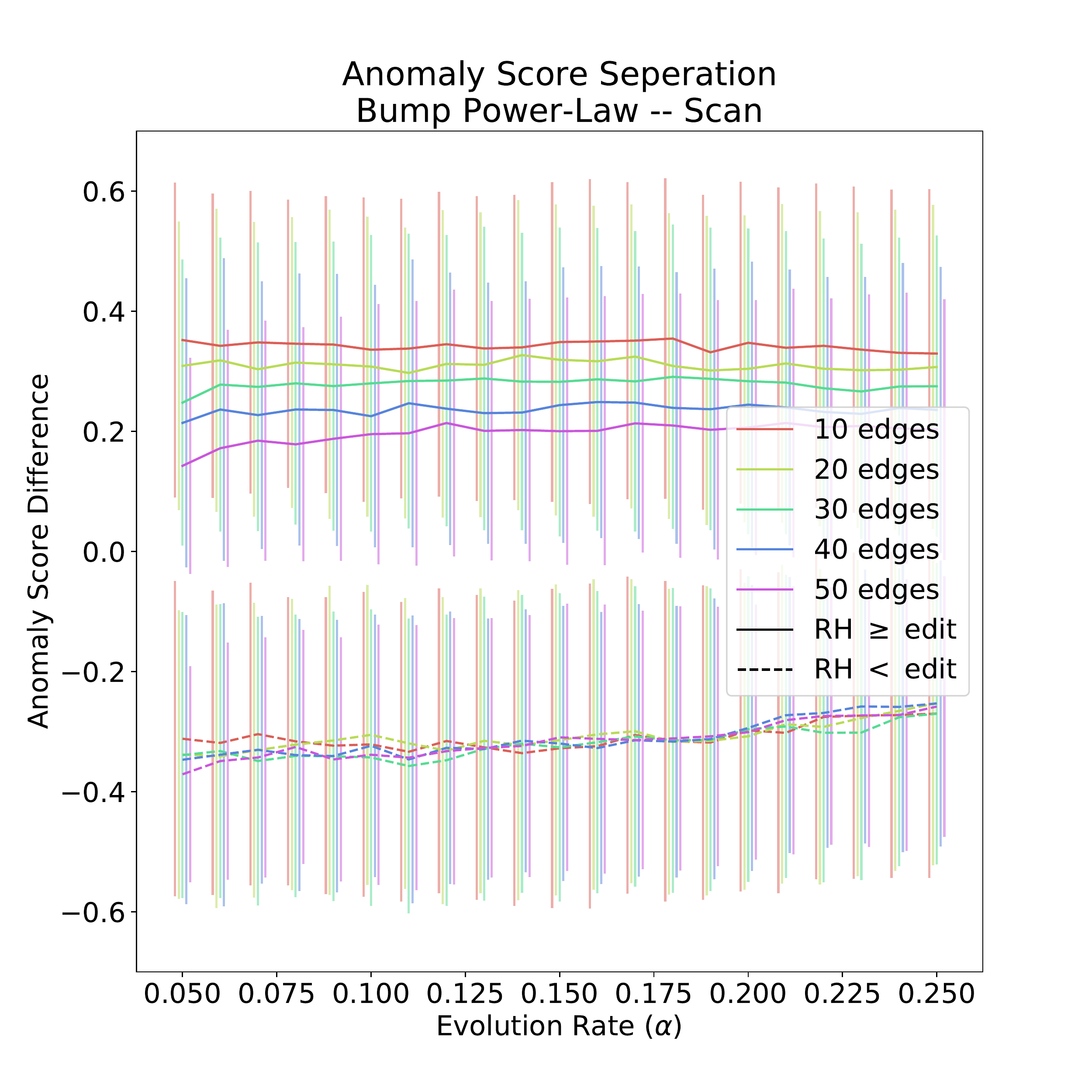}
  }
  \hfill
  \subfloat[][Bump Power-Law \\ Lateral movement]{
    \includegraphics[trim = 15 28 69 85, clip, width = .4\textwidth]{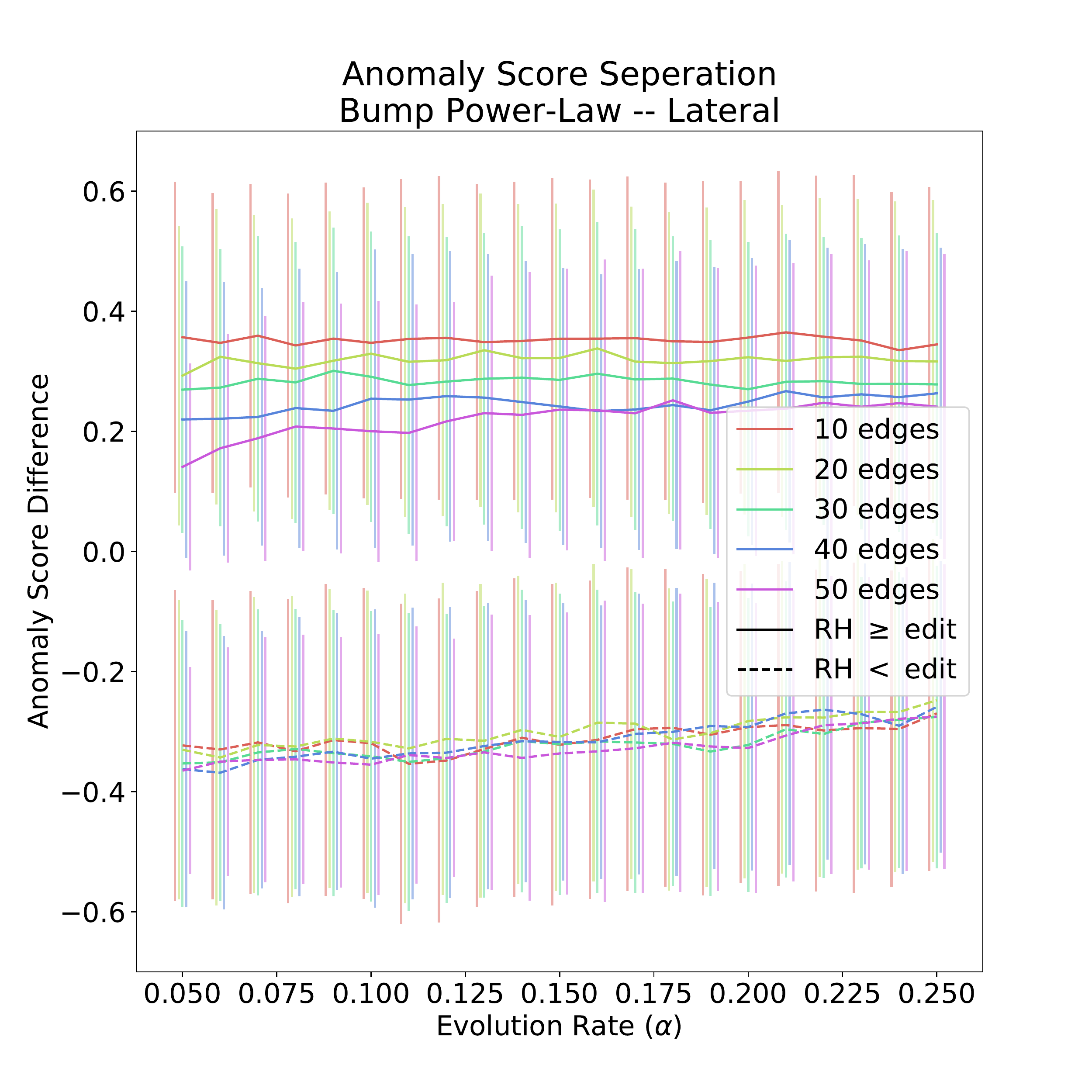}
  }
  \hfill \phantom{}
  \caption{Difference between edit distance and RH distance}\label{F:edit_H2Hrange}
  \end{figure}

\section{Conclusions}

In this work, we conducted an experimental and statistical study of Relative Hausdorff distance in the context of time-evolving sequences of graphs. 
Applying RH distance as an anomaly detection tool, we first tested its detection of red team events across multiple modalities in real cyber security data. 
We found evidence that RH distance values register statistical change-points at red team events, although these results were sensitive to window length, the granularity of pairwise RH measurements, and subject to the limitations of the data. 
In order to test RH distance in a more controlled and rigorous manner, we then turned our attention to a temporal graph model inspired by cyber-data.

Using this temporal graph model to generate synthetic sequences of evolving graphs, we experimentally tested the sensitivity of RH distance to two attack profiles. 
To broaden our tests scope, we considered a multitude of parameter settings in which we varied the input degree distribution, temporal evolution rate, and intensity of the attack signal. 
In its own right, RH distance performed respectably, yielding ROC curves above the line of no-discrimination for every scenario tested. 
Compared with other similarity measures, RH distance consistently outperformed another lightweight similarity measure based on Kolmogorov-Smirnov distance, while its performance against the computationally-intensive edit distance was more mixed: while edit distance clearly outperformed RH distance under scenarios featuring the power law degree distribution, RH distance was better able to detect the low-intensity lateral movement attack under the bump power law degree distribution. 

Anomaly detection generally, and even specifically in cyber security, is not amenable to a ``one method to rule them all'' mentality. 
Indeed, there are many types of anomalies and one does not expect them to all be caught by the same detector.
It is important to recognize that our analysis does not use all of the available information pertinent to real cyber data. 
Distilling a time interval of data down to a single graph and removing all metadata is likely to introduce many false positives.
It could be that a graph is anomalous given the recent context, but the behavior is fully expected by cyber security operations analysts (e.g. a daily backup may appear to be an exfiltration if the IP addresses involved aren't considered). 
In the other direction, if the graph does not contain the metadata that would flag an anomaly, this may similarly introduce false negatives. 
By integrating metadata into our analysis, it is possible that as anomalies are discovered, this metadata could be used to help classify them as benign and nefarious anomalies. 
Lastly, it is also worth noting that our analyses considered the entire data in a given period, as opposed to an online approach. Whether and how RH distance might be utilized in online anomaly detection frameworks remains another open topic for future research. \\

%
%

\noindent{\bf Funding} \\
\noindent PNNL is operated by Battelle for the United States Department of Energy under Contract DE-AC05-76RL01830. 
\textit{PNNL Information Release:} PNNL-SA-141621.
The design of the study and all data analysis and interpretation was done under the direction of the authors of this paper. \\


\noindent{\bf Acknowledgements}\\
The authors would like to thank our colleague Carlos Ortiz-Marrero for reviewing the manuscript.

\bibliographystyle{siam}
\bibliography{RH_refs}
\end{document}